\documentclass[11pt]{article}

\usepackage{fullpage}
\usepackage{soul}
\usepackage[utf8]{inputenc}
\usepackage[small]{caption}
\usepackage{xcolor}
\usepackage{amsmath}
\usepackage{booktabs}
\usepackage{enumitem}
\usepackage{framed}
\usepackage{tikz}
\usepackage{url}

\usepackage{amsthm}
\usepackage{amssymb}
\usepackage{afterpage}

\usepackage{algpseudocode,graphicx,caption,subcaption,xspace,nicefrac,mathtools,bm}

\usepackage[ruled,vlined,linesnumbered,norelsize]{algorithm2e}

\algnotext{EndFor}
\algnotext{EndIf}
\algnotext{EndWhile}
\algnotext{EndProcedure}

\newtheorem{theorem}{Theorem}[section]
\newtheorem{corollary}[theorem]{Corollary}
\newtheorem{proposition}[theorem]{Proposition}
\newtheorem{lemma}[theorem]{Lemma}

\theoremstyle{definition}
\newtheorem{definition}[theorem]{Definition}
\newtheorem{example}[theorem]{Example}

\usepackage{natbib}
\usepackage{cleveref}

\newcommand{\ot}{\leftarrow}

\allowdisplaybreaks

\usepackage{authblk}

\title{\bf Fair Division with Two-Sided Preferences\footnote{A preliminary version of this work \citep{IgarashiKaSu23} appeared in Proceedings of the 32nd International Joint Conference on Artificial Intelligence (IJCAI 2023).}}

\author[1]{Ayumi Igarashi} 
\author[1]{Yasushi Kawase}
\author[2]{Warut Suksompong}
\author[3]{Hanna Sumita}

\affil[1]{University of Tokyo}
\affil[2]{National University of Singapore}
\affil[3]{Tokyo Institute of Technology}

\date{\vspace{-1cm}}

\begin{document}

\maketitle

\begin{abstract}
We study a fair division setting in which participants are to be fairly distributed among teams, where not only do the teams have preferences over the participants as in the canonical fair division setting, but the participants also have preferences over the teams.
We focus on guaranteeing envy-freeness up to one participant (EF1) for the teams together with a stability condition for both sides.
We show that an allocation satisfying EF1, swap stability, and individual stability always exists and can be computed in polynomial time, even when teams may have positive or negative values for participants.
When teams have nonnegative values for participants, we prove that an EF1 and Pareto optimal allocation exists and, if the valuations are binary, can be found in polynomial time.
We also show that an EF1 and justified envy-free allocation does not necessarily exist, and deciding whether such an allocation exists is computationally difficult.
\end{abstract}

\section{Introduction}
\label{sec:intro}

The new season of a youth sports league is starting in three months, and the league organizers need to allocate the participating players to the available teams.
How can they accomplish this task in a satisfactory way, so that all parties involved can look forward to the upcoming season instead of grumbling about the allocation?

A principal consideration in such allocation tasks is fairness, and the problem of fairly dividing resources (in this case, participants) among interested recipients (in this case, teams) has long been studied in economics under the name of \emph{fair division} \citep{BramsTa96,Moulin03,Moulin19}. 
Ensuring fairness among teams is crucial for the sustainability of the league, the motivation for taking part, and the enhancement of competition.
Among the fairness notions that have been proposed in the literature, one of the strongest and most prominent is \emph{envy-freeness}, which stipulates that no team should envy another team based on the sets of participants that they receive.\footnote{Due to the setting that we study, throughout this paper we will use the terms \emph{team} and \emph{participant} in place of the standard fair division terms \emph{agent} and \emph{item}, respectively.}
Even though an envy-free allocation may not exist (e.g., if there is one highly-coveted superstar), an intuitive relaxation called \emph{envy-freeness up to one participant (EF1)}---that is, any envy that one team has toward another team can be eliminated upon the removal of some participant in the envied team---can always be fulfilled \citep{LiptonMaMo04,Budish11}.
Another relevant criterion is \emph{balancedness}, which requires the participants to be distributed as equally among the teams as possible.\footnote{One could view balancedness as EF1 with respect to the number of allocated participants.}
Balancedness can be especially desirable when allocating players to sports teams, as each team may need to have a fixed number of players due to the rules of the sport.
Assuming that teams have additive and nonnegative values for participants, an allocation that is both EF1 and balanced always exists and can be found efficiently via a simple round-robin algorithm (see, e.g., \citealt[p.~7]{CaragiannisKuMo19}).
In fact, this algorithm forms the basis of \emph{draft} processes used in many sports leagues around the world.\footnote{Please refer to \url{http://wikipedia.org/wiki/Draft_(sports)}.}

While EF1 provides a strong fairness guarantee with respect to the teams' preferences, it overlooks the fact that the participating players may have preferences over the teams as well, for example, depending on their familiarity with the team managers or the proximity of their residence to the training grounds.
Clearly, ignoring the preferences of the participants may lead to a suboptimal allocation.
As an extreme case, if every team is indifferent between all participants, then swapping a pair of participants keeps the teams as happy as before and may make both of the swapped participants much happier.
In addition to our sports league example, two-sided preferences also occur in the allocation of students to thesis supervisors, employees to different branches of a restaurant chain, or volunteers to community service clubs.
Moreover, the participant preferences could signify the \emph{suitability} of the teams for the participants---for instance, the ``participants'' could represent tasks (such as household chores or papers to be reviewed), and the ``teams'' have varying levels of ability to perform the tasks.
As we shall discuss in \Cref{sec:relatedwork}, with few exceptions, the vast fair division literature deals exclusively with one-sided preferences.
Can we find an allocation that is fair to the teams and, at the same time, satisfies a stability condition with respect to the preferences of both sides?

\subsection{Our Results}

As is common in fair division, we assume that the teams have additive valuations over the participants, that is, a team's value for a set of participants is equal to the sum of its values for the individual participants in this set.\footnote{As \citet{CaragiannisKuMo19} noted, the assumption of additive valuations is ``indispensable'' in practical applications, because it permits simple elicitation of the valuations.}
Some of our results allow these values to be either positive or negative; this corresponds to the allocation of \emph{indivisible goods and chores} \citep{AzizCaIg22}.
For consistency of terminology, we will use the terms \emph{nonnegative-value participants} and \emph{nonpositive-value participants} instead of \emph{goods} and \emph{chores}, respectively.
We formally define the notions that we consider and outline some relationships between them in \Cref{sec:prelim}.

In \Cref{sec:swap-stability}, we focus on \emph{swap stability}, that is, the requirement that no swap between two participants makes at least one of the four involved parties better off and none of them worse off.
First, we observe that even with nonnegative-value participants, starting with an arbitrary EF1 allocation and letting participants make beneficial swaps may result in an allocation that violates EF1.
Despite this fact, for teams with arbitrary (positive or negative) values over participants, we present a polynomial-time algorithm that produces a balanced and swap stable allocation satisfying \emph{EF[1,1]}, a relaxation of EF1 where one participant may be removed from each of the envying team and the envied team.\footnote{When both positive and negative values are allowed, EF1 permits one participant to be removed from either the envying team or the envied team (but not both) \citep{AzizCaIg22}.}
Since EF[1,1] reduces to EF1 for nonnegative-value participants as well as for nonpositive-value participants, we obtain the same result for EF1 in each of these cases.
We then note two ways in which our arbitrary-value result cannot be improved: EF[1,1] cannot be strengthened to EF1, and we cannot simultaneously attain \emph{individual stability}---the condition that no deviation of a participant to another team makes the participant better off and neither of the involved teams worse off.
Nevertheless, we show that if we give up balancedness, both of these improvements become possible: an allocation that satisfies EF1, swap stability, and individual stability exists and can be found efficiently.

\renewcommand{\arraystretch}{1.2}
\begin{table}[!t]
\centering
    \begin{tabular}{| c | c |}
    \hline
     \textbf{Properties} & \textbf{Existence} \\ \hline \hline
     EF[1,1] + balanced + swap stable & Yes (Thm.~\ref{thm:balanced})  \\ \hline
     EF1 + balanced & No (Prop.~\ref{prop:balanced-EF1})  \\ \hline
     balanced + individually stable & No (Prop.~\ref{prop:balanced-IS})   \\ \hline
     EF1 + swap stable + individually stable & Yes (Thm.~\ref{thm:individual-stable})   \\ \hline \hline
     EF1 + PO (nonnegative-value participants) & Yes (Thm.~\ref{thm:MNW}) \\ \hline
     EF1 + PO + team-PO (two teams) & Yes (Thm.~\ref{thm:gaw})   \\ \hline
     EF1 + participant-PO & No (Prop.~\ref{prop:EF1-participant-PO}) \\ \hline \hline
     EF1 + justified EF & No (Prop.~\ref{prop:justifiedEF-counterexample}) \\ \hline
    \end{tabular}
    \caption{Summary of our results on whether each combination of properties can always be satisfied simultaneously, with the corresponding theorem or proposition number.}
    \label{table:summary}
\end{table}

Next, in \Cref{sec:PO}, we consider the notion of \emph{Pareto optimality (PO)}---no allocation can make a party (i.e., either participant or team) better off without making another party worse off---which is stronger than both swap stability and individual stability.
We prove that deciding whether an allocation is PO or not is coNP-complete even for two teams with identical valuations, nonnegative-value participants, and a balanced allocation.
On the other hand, for two teams with arbitrary valuations, we show that an extension of the \emph{generalized adjusted winner procedure} of \citet{AzizCaIg22} computes an EF1 and PO allocation in polynomial time.
For any number of teams and nonnegative-value participants, we observe that an EF1 and PO allocation always exists.
Moreover, we demonstrate that such an allocation can be found efficiently in two special cases: (i) the teams have binary valuations over the participants, and (ii) there are three teams with identical valuations, and each participant has a favorite team and is indifferent between the other two teams.
We also provide a pseudopolynomial-time algorithm when the number of teams is constant.

Finally, in \Cref{sec:justifiedEF}, we examine \emph{justified envy-freeness}, a stability notion from the two-sided matching literature: participant~$p_i$ is said to have justified envy toward another participant~$p_j$ assigned to team~$k$ if $k$ prefers $p_i$ to $p_j$ and $p_i$ prefers $k$ to her assigned team.
Perhaps surprisingly, we show that an EF1 and justified envy-free allocation may not exist, even for two teams and nonnegative-value participants who all prefer the same team.
We then prove that deciding whether such an allocation exists is NP-complete even for nonnegative-value participants who have strict preferences over the teams.
On the other hand, if one adds the condition that there are two teams, we show that the problem becomes polynomial-time solvable.

Our (non-)existence results are summarized in \Cref{table:summary}.
In \Cref{app:quota}, we demonstrate how some of our notions and results can be extended to accommodate quota constraints.

\subsection{Related Work}
\label{sec:relatedwork}

Even though fair division has given rise to a sizable body of work \citep{BramsTa96,Moulin03,Moulin19}, the vast majority of the literature assumes one-sided preferences---in our terminology, the teams have preferences over the participants, but not vice versa.
A small number of recent papers have combined fairness concepts with two-sided preferences.
\citet{FreemanMiSh21} considered many-to-many matching and proposed the notion of \emph{double envy-freeness up to one match (DEF1)}, which requires EF1 to hold for both sides simultaneously.
Note that in our many-to-one setting, DEF1 is meaningless on the participant side because it is always trivially satisfied.
\citet{GollapudiKoPl20} studied many-to-many matching in a dynamic setting; their positive results primarily hold for symmetric binary valuations, which are much more restrictive than the valuations that we allow.
\citet{PatroBiGa20} investigated fairness in two-sided platforms between producers and customers, but assumed that producers are indifferent between customers.
We emphasize that none of these papers addressed a model that suitably captures our motivating examples such as the allocation of sports players to teams or volunteers to community service clubs.

While most fair division papers assume that the items (in our terminology, participants) are goods and some assume that they are chores, a recent line of work has relaxed these assumptions by allowing items to be either goods or chores, with this evaluation possibly varying across agents (in our terminology, teams) \citep{AleksandrovWa20,AzizRe20,BhaskarSrVa21,KulkarniMeTa21,AzizCaIg22,HosseiniSiVa23,BercziBeBo24}.
\citet{BogomolnaiaMoSa17} coined the term \emph{mixed manna} to describe such division problems.
\citet{AzizCaIg22} showed that an EF1 allocation always exists and can be found efficiently, and that the same is true for an EF1 and PO allocation when there are two agents.
\citet{BhaskarSrVa21} provided a polynomial-time algorithm for computing an EF1 allocation when valuations are not necessarily additive but each agent can partition the items into those that always yield nonnegative marginal utility and those that always yield nonpositive marginal utility.
\citet{BercziBeBo24} established the existence of an EF1 allocation for arbitrary set valuations in the case of two agents.

Finally, even though our setting can be seen as a many-to-one matching with two-sided preferences, to the best of our knowledge, the matching literature has not considered fairness among teams, whether using EF1 or other fairness notions.
By contrast, justified envy is commonly investigated in two-sided matching.
Indeed, it forms the basis of the stability notion in one-to-one matching famously studied by \citet{GaleSh62}.
Justified envy-freeness is also often considered in many-to-one matching \citep{AbdulkadirogluSo03,FragiadakisIwTr15,KamadaKo17,WuRo18,AbdulkadirogluChPa20,Yokoi20}.
In particular, in many-to-one matching with upper quotas (also known as the \emph{Hospitals/Residents problem}), it is viewed as a relaxation of stability: while stability disallows the existence of a participant who wants to claim a team's vacant seat, justified envy-freeness tolerates the existence of such a participant.
Thus, a justified envy-free matching may exist even when a stable matching does not.

\section{Preliminaries}
\label{sec:prelim}

For each positive integer $z$, let $[z] \coloneqq \{1,\dots,z\}$.
Let $T = [n]$ be the set of teams and $P = \{p_1,\dots,p_m\}$ the set of participants; we sometimes refer to either a team or a participant as a \emph{party}.\footnote{While we formulate our model in the language of team allocation, as we discussed in \Cref{sec:intro}, our model can be applied to a wide range of scenarios.}
Each participant $p\in P$ has a weak transitive preference $\succsim_p$ over the teams, with $\succ_p$ and $\sim_p$ denoting the strict and equivalence part of $\succsim_p$, respectively.\footnote{All notions considered in this paper take into account only the participants' ordinal preferences, so we do not assume cardinal utilities for the participants.}
The \emph{rank} of a team $i$ for a participant $p$ is defined as $1$ plus the number of teams $j$ such that $j\succ_p i$.
Each team $i\in T$ has a valuation function (or utility function) $v_i\colon 2^P\to\mathbb{R}$ over subsets of participants.
We assume that the valuations are additive, i.e., $v_i(P') = \sum_{p\in P'}v_i(\{p\})$ for all $i\in T$ and $P'\subseteq P$.
For convenience, we write $v_i(p)$ instead of $v_i(\{p\})$.
An \emph{instance} consists of the teams and participants, as well as the valuations and preferences of both sides.
Sometimes we will consider the setting of \emph{nonnegative-value participants} (resp., \emph{nonpositive-value participants}), which means that $v_i(p) \ge 0$ (resp., $v_i(p) \le 0$) for all $i\in T$ and $p\in P$.

An \emph{allocation} $A=(A_1,A_2,\ldots,A_n)$ is an ordered partition of $P$ into $n$ parts, where the part~$A_i$ is assigned to team~$i$.
We will investigate several fairness and stability notions for allocations.
A basic fairness consideration on the team side is (almost) envy-freeness.

\begin{definition}
\label{def:EF1}
An allocation $A$ is said to satisfy
\begin{itemize}
\item \emph{EF1} if for all distinct $i,j\in T$, it holds that $v_i(A_i\setminus X)\ge v_i(A_j\setminus Y)$ for some $X\subseteq A_i$ and $Y\subseteq A_j$ with $|X\cup Y|\le 1$;
\item \emph{EF[1,1]} if for all distinct $i,j\in T$, it holds that $v_i(A_i\setminus X)\ge v_i(A_j\setminus Y)$ for some $X\subseteq A_i$ and $Y\subseteq A_j$ with $|X|,|Y|\le 1$.
\end{itemize}
\end{definition}

EF1 was first studied for nonnegative-value participants by \citet{LiptonMaMo04} (though the term itself was coined by \citet{Budish11}) and subsequently extended to arbitrary-value participants by \citet{AzizCaIg22}, 
while EF[1,1] was recently introduced by \citet{ShoshanSeHa23}.
It follows immediately from the definition that EF1 implies EF[1,1].
Moreover, if all participants yield nonnegative value, there is no reason to remove a participant from $A_i$, so EF1 and EF[1,1] coincide in this case; an analogous statement holds for nonpositive-value participants with $A_i$ replaced by $A_j$.

Our next criterion is balancedness, which requires the participants to be distributed as equally among the teams as possible.

\begin{definition}
An allocation $A$ is said to be \emph{balanced} if $\big||A_i|-|A_j|\big|\le 1$ for all $i,j\in T$.
\end{definition}

Observe that if there exists a constant $c\ne 0$ such that $v_i(p) = c$ for all $i\in T$ and $p\in P$, then both EF1 and EF[1,1] coincide with balancedness.

We now define stability concepts, several of which take into account the preferences of both sides.
We say that a party is \emph{better off} (resp., \emph{worse off}) if it receives a better (resp., worse) outcome with respect to its valuation function (for a team) or preference (for a participant).

\begin{definition}
Given an allocation $A$, a swap between participants $p\in A_i$ and $p'\in A_j$ (for some $i,j\in T$) is a \emph{beneficial swap} if it makes at least one of the four involved parties better off and none of them worse off. 
A deviation of a participant~$p$ to another team is a \emph{beneficial deviation} if it makes $p$ better off and neither of the teams involved worse off.

An allocation $A$ is said to be \emph{swap stable} if it does not admit a beneficial swap.\footnote{One could define a stronger version of swap stability, where a swap is beneficial if it makes at least one of the two involved participants better off and neither of them worse off (without taking into account the team preferences). 
However, this notion is clearly incompatible with EF1 or EF[1,1].
For example, this is the case in the instance with $n = 2$ teams and $m = 4$ participants, where $v_1(p_1)=v_1(p_2)=v_2(p_3)=v_2(p_4) = 1$, $v_1(p_3) = v_1(p_4) = v_2(p_1) = v_2(p_2) = 0$, $2\succ_{p_j}1$ for $j\in\{1,2\}$, and $1\succ_{p_j}2$ for $j\in\{3,4\}$.}
It is said to be \emph{individually stable} if it does not admit a beneficial deviation.\footnote{This is analogous to the notion of \emph{contractual individual stability} in hedonic games \citep{AzizSa16}. 
If we only require that the deviation does not make the participant's new team worse off (as in \emph{individual stability} in hedonic games), then supposing that all participants yield nonnegative value, the only stable allocations are the allocations in which every participant is assigned to one of her most preferred teams.}
\end{definition}

\begin{definition}
An allocation $A$ is said to be \emph{Pareto dominated} by another allocation $A'$ if no party is worse off in $A'$ than in $A$ and at least one party is better off; in this case, $A'$ is a \emph{Pareto improvement} of $A$.
An allocation $A$ is \emph{Pareto optimal (PO)} if it is not Pareto dominated by any other allocation.

We define \emph{team-Pareto dominated}, \emph{team-Pareto optimal (team-PO)}, \emph{participant-Pareto dominated}, and \emph{participant-Pareto optimal (participant-PO)} similarly, with ``party'' replaced by ``team'' and ``participant'', respectively.
\end{definition}

Although PO clearly implies both swap stability and individual stability, it implies neither team-PO nor participant-PO.

\begin{proposition}
\label{prop:PO-team-PO}
PO does not necessarily imply team-PO.
\end{proposition}

\begin{proof}
Consider the following instance with $n = m = 2$:
\begin{itemize}
\item $v_1(p_1) = v_1(p_2) = v_2(p_1) = 1$ and $v_2(p_2) = 0$;
\item $1\succ_{p_1} 2$ and $2\succ_{p_2} 1$.
\end{itemize}
The allocation $A = \big(\{p_1\}, \{p_2\}\big)$ 
is team-Pareto dominated by the allocation $A'=\big(\{p_1,p_2\},\emptyset\big)$, 
so $A$ is not team-PO.
However, $A$ is PO since each participant is already assigned to her unique favorite team.
\end{proof}

\begin{proposition}
\label{prop:PO-participant-PO}
PO does not necessarily imply participant-PO.
\end{proposition}

\begin{proof}
Consider the following instance with $n = m = 2$:
\begin{itemize}
\item $v_1(p_1) = v_1(p_2) = v_2(p_1) = v_2(p_2) = 1$;
\item $1\succ_{p_1} 2$ and $1\succ_{p_2} 2$.
\end{itemize}
The allocation $A = \big(\{p_1\}, \{p_2\}\big)$  
is participant-Pareto dominated by the allocation $A'=\big(\{p_1,p_2\},\emptyset\big)$,  
so $A$ is not participant-PO.
However, one can check that $A$ is PO.
\end{proof}

On the other hand, PO is implied by the combination of team-PO and participant-PO.

\begin{proposition}
Team-PO and participant-PO together imply PO.
However, either team-PO or participant-PO alone does not necessarily imply PO.
\end{proposition}

\begin{proof}
Consider an allocation $A$ that is both team-PO and participant-PO, and assume for contradiction that it is not PO.
Then, there is an allocation $A'$ that is a Pareto improvement of $A$.
Every party is better off in $A'$ compared to in $A$, and at least one party is strictly better off.
If a team is strictly better off, then $A'$ team-Pareto dominates $A$, so $A$ is not team-PO.
If a participant is strictly better off, then $A'$ participant-Pareto dominates $A$, so $A$ is not participant-PO.
In either case, we arrive at a contradiction.

To see that team-PO does not imply PO, consider the following instance with $n = m = 2$:
\begin{itemize}
\item $v_1(p_1) = v_1(p_2) = v_2(p_1) = v_2(p_2) = 1$;
\item $1\succ_{p_1} 2$ and $2\succ_{p_2} 1$.
\end{itemize}
The allocation $A = \big(\{p_2\}, \{p_1\}\big)$  
is Pareto dominated by the allocation $A'=\big(\{p_1\},\{p_2\}\big)$,  
so $A$ is not PO.
However, one can check that $A$ is team-PO.

To see that participant-PO does not imply PO, consider the following instance with $n = m = 2$:
\begin{itemize}
\item $v_1(p_1) = v_2(p_2) = 1$ and $v_2(p_1) = v_1(p_2) = 0$;
\item $1\sim_{p_1} 2$ and $2\sim_{p_2} 1$.
\end{itemize}
The allocation $A = \big(\{p_2\}, \{p_1\}\big)$  
is Pareto dominated by the allocation $A'=\big(\{p_1\},\{p_2\}\big)$,  
so $A$ is not PO.
However, $A$ is participant-PO since both participants are indifferent between both teams.
\end{proof}

Finally, we define the concept of justified envy.

\begin{definition}
Given an allocation $A$, a participant~$p\in A_i$ is said to have \emph{justified envy} toward a participant~$p'\in A_j$ if $j\succ_p i$ and $v_j(p) > v_j(p')$.
An allocation $A$ is \emph{justified envy-free (justified EF)} if no participant has justified envy toward another participant.
\end{definition}

One could consider a weaker version of justified envy where the envy is considered justified even if the team is indifferent (i.e., $j\succ_p i$ and $v_j(p) \ge v_j(p')$); this leads to a stronger version of justified EF. 
While our non-existence result (Proposition~\ref{prop:justifiedEF-counterexample}) automatically carries over to this version, the alternative definition appears to be too strong in the sense that even in a simple instance with two teams having value $1$ for each of two participants, if both participants prefer the same team, then no EF1 and (alternative) justified EF allocation exists.
In particular, our existence result (\Cref{thm:justifiedEF-algo-two-iden}) does not hold for this version.\footnote{On the other hand, our algorithm in \Cref{thm:justifiedEF-algo-two} can be adapted to this version of justified EF.}
Moreover, our definition of justified envy corresponds to the notion of \emph{weak stability} from the stable matching literature, whereas the alternative definition does not correspond to other known stability notions such as \emph{strong stability} and \emph{super-stability} (see, e.g., the book by \citet[p.~27]{Manlove13} for details).

We end this section by showing that there are no implication relations between justified EF and the three notions PO, swap stability, and individual stability.
\begin{itemize}
\item \emph{Justified EF does not imply individual stability.}
Consider an instance with $n=2$, $m=1$, $v_1(p_1) = 1$, $v_2(p_1) = 0$, and $1\succ_{p_1}2$.
The allocation $A = (\emptyset, \{p_1\})$ is justified EF, but it is not individually stable. 
This also means that justified EF does not imply PO.
\item \emph{Justified EF does not imply swap stability.}
Consider an instance with $n = m = 2$, $v_1(p_1) = v_2(p_2) = 1$, $v_1(p_2) = v_2(p_1) = 0$, $1\sim_{p_1}2$, and $1\sim_{p_2}2$.
The allocation $A = (\{p_2\}, \{p_1\})$ is justified EF, but it is not swap stable.
\item \emph{PO does not imply justified EF.}
This follows from the fact that an EF1 and PO allocation exists for every instance with nonnegative-value participants (\Cref{thm:MNW}) but there is an instance with nonnegative-value participants in which no EF1 and justified EF allocation exists (Proposition~\ref{prop:justifiedEF-counterexample}).
This also means that neither swap stability nor individual stability implies justified EF.
\end{itemize}

\section{Swap Stability}
\label{sec:swap-stability}

In this section, we focus on swap stability.
A natural idea for obtaining an EF1 and swap stable allocation is to start with an arbitrary EF1 allocation and let participants swap as long as a beneficial swap exists.
Note that determining whether beneficial swaps exist (and, if so, finding such a swap) can be done in polynomial time since we can simply check all pairs of participants.
However, as can be seen in the following example, this approach does not always result in an EF1 allocation, even for nonnegative-value participants.

\begin{example}\label{ex:swapBreakEF1}
Consider the following instance with $n = 3$ and $m = 6$:
\begin{itemize}
\item $v_i(p_j) = 0$ for $i\in[2]$ and $j\in [6]$;
\item $v_3(p_1) = v_3(p_2) = 1$ and $v_3(p_3) = v_3(p_4) = v_3(p_5) = v_3(p_6) = 0$;
\item each participant has a unique favorite team and is indifferent between the other two teams: $p_1$ and $p_2$ prefer team~$1$, $p_4$ and $p_5$ prefer team~$2$, and $p_3$ and $p_6$ prefer team~$3$.
\end{itemize}
The allocation $A = \big(\{p_1,p_4\},\{p_2,p_5\},\{p_3,p_6\}\big)$  
is EF1.
The swap between $p_2$ and $p_4$ is the unique beneficial swap; let $A'$ be the allocation after this swap.
The allocation $A'$ is swap stable, but it is not EF1 because team~$3$ envies team~$1$ by more than one participant.
\end{example}

In spite of this example, we show that an EF[1,1] and swap stable allocation that is moreover balanced always exists and can be found efficiently.

\begin{theorem}\label{thm:balanced}
For any instance, a balanced allocation that satisfies EF[1,1] and swap stability exists and can be computed in polynomial time. 
\end{theorem}

Since EF[1,1] reduces to EF1 for nonnegative-value participants as well as for nonpositive-value participants, \Cref{thm:balanced} implies the following corollary.

\begin{corollary}
\label{cor:balanced}
For any nonnegative-value participant instance, a balanced allocation that satisfies EF1 and swap stability exists and can be computed in polynomial time. 
The same holds for any nonpositive-value participant instance.
\end{corollary}

\begin{algorithm*}[!tb]
  \caption{For computing an EF[1,1], swap stable, and balanced allocation}\label{alg:swap-stable-rr}
  Construct a complete bipartite graph $G=(Q,P; E)$ with weight function $w\colon E\to\mathbb{R}$ where $Q=[m]$ and $w(q,p)=v_{f(q)}(p)$ for each $(q,p)\in Q\times P$; \quad \emph{// The set $Q$ mimics the available positions within the teams}\\ 
  Compute a perfect matching $\mu\subseteq Q\times P$ such that the weight of the edge adjacent to vertex $1\in Q$ is as large as possible, and subject to this condition, the weight of the edge adjacent to vertex $2\in Q$ is as large as possible, and so on until vertex $m\in Q$; \quad \quad\quad\quad\quad  \emph{// Simulate round-robin with only the participants' values being assigned to teams}\\ \label{line:perfect-matching}
  Let $E^*=\{(q,p)\in Q\times P \mid w(q,p)=w(q,\mu_q)\}$; \quad \emph{// Keep edges that are as good for the teams as those in $\mu$}\\
  Compute a perfect matching $\mu^*$ in $G^*=(Q,P; E^*)$ such that the sum over all participants $p\in P$ of the rank of team~$f(\mu^*_p)$ for participant~$p$ is minimized\;\label{line:matching-special}
  Return the allocation $A$ such that $p$ is allocated to team~$f(q)$ for each $(q,p)\in \mu^*$\;
\end{algorithm*}

Our algorithm for \Cref{thm:balanced} proceeds in a round-robin manner.
However, instead of assigning a participant to a team in each turn as is usually done, we only assign a participant's \emph{value} to the team; this ensures that more possibilities are available in later turns.
Then, among the allocations that satisfy the determined values for teams, we compute an allocation that minimizes the sum of the participants' ranks for the teams.
Formally, the algorithm is described as Algorithm~\ref{alg:swap-stable-rr}.
For each positive integer~$q$, we denote by $f(q)$ the unique integer in~$[n]$ such that $f(q)\equiv q\pmod{n}$.
For a matching~$\mu$ with $(q,p)\in \mu$, we define the notation $\mu_q$ and $\mu_p$ so that $\mu_q = p$ and $\mu_p = q$.
Note that each $q\in Q$ corresponds to a copy of team~$f(q)$.
For example, given an instance with $n=3$ and $m=7$, we construct a bipartite graph $G$ with $|Q| = |P| = 7$.
The vertices $1,4,7\in Q$ correspond to copies of team~$1$, the vertices $2,5\in Q$ correspond to copies of team~$2$, and the vertices $3,6\in Q$ correspond to copies of team~$3$.

It is clear that the allocation produced by \Cref{alg:swap-stable-rr} is balanced.
To establish \Cref{thm:balanced}, we prove the remaining properties of the algorithm, including its polynomial running time, in the following three lemmas.

\begin{lemma}
\label{lem:EF1-1}
The output allocation $A$ of Algorithm~\ref{alg:swap-stable-rr} is EF[1,1].
\end{lemma}
\begin{proof}
By definition of EF[1,1], we need to show that, for all distinct $i,j\in T$, we have $v_i(A_i\setminus X)\ge v_i(A_j\setminus Y)$ for some $X\subseteq A_i$ and $Y\subseteq A_j$ with $|X|,|Y|\le 1$.
The statement holds trivially if $m\le n$ since each team receives at most one participant, so assume that $m > n$.
Fix arbitrary distinct $i,j\in T$.
We consider three cases based on the sizes of $A_i$ and $A_j$.
(In what follows, $\mu$ and $\mu^*$ are as defined in Algorithm~\ref{alg:swap-stable-rr}.)

First, suppose that $|A_i|=|A_j|$. Let $k \coloneqq|A_i| \ge 1$.
Then, 
\begin{align*}
    v_i(A_i\setminus\{\mu^*_{n(k-1)+i}\})
    &=   \sum_{\ell=1}^{k-1} v_i(\mu^*_{n(\ell-1)+i})\\
    &=   \sum_{\ell=1}^{k-1} v_i(\mu_{n(\ell-1)+i})\\
    &\ge   \sum_{\ell=1}^{k-1} v_i(\mu_{n\ell+j})
    = \sum_{\ell=2}^{k} v_i(\mu_{n(\ell-1)+j})
    = \sum_{\ell=2}^{k} v_i(\mu^*_{n(\ell-1)+j})
    = v_i(A_j\setminus\{\mu^*_{j}\}),
\end{align*}
where $v_i(\mu_{n(\ell-1)+i}) \ge v_i(\mu_{n\ell + j})$ holds because otherwise the weight of the edge in $\mu$ adjacent to vertex $n(\ell-1)+i \in Q$ can be increased without decreasing the weights of the edges adjacent to vertices $1,2,\dots,n(\ell-1)+i-1 \in Q$, contradicting the definition of $\mu$.

Next, suppose that $|A_i| > |A_j|$; in particular, it must be that $i < j$.
Let $|A_i|=k$ and $|A_j|=k-1$.
Applying a similar argument as in the case $|A_i| = |A_j|$, we have
\begin{align*}
    v_i(A_i\setminus\{\mu^*_{n(k-1)+i}\})
    &=   \sum_{\ell=1}^{k-1} v_i(\mu^*_{n(\ell-1)+i})\\
    &=   \sum_{\ell=1}^{k-1} v_i(\mu_{n(\ell-1)+i})
    \ge \sum_{\ell=1}^{k-1} v_i(\mu_{n(\ell-1)+j})
    = \sum_{\ell=1}^{k-1} v_i(\mu^*_{n(\ell-1)+j})
    = v_i(A_j).
\end{align*}

Finally, suppose that $|A_i| < |A_j|$; in particular, it must be that $i > j$.
Let $|A_i|=k-1$ and $|A_j|=k$.
Applying a similar argument once more, we have
\begin{align*}
    v_i(A_i)
    &=   \sum_{\ell=1}^{k-1} v_i(\mu^*_{n(\ell-1)+i})\\
    &=   \sum_{\ell=1}^{k-1} v_i(\mu_{n(\ell-1)+i})\\
    &\ge \sum_{\ell=1}^{k-1} v_i(\mu_{n\ell+j})
    = \sum_{\ell=2}^{k} v_i(\mu_{n(\ell-1)+j})
    = \sum_{\ell=2}^{k} v_i(\mu^*_{n(\ell-1)+j})
    = v_i(A_j\setminus\{\mu^*_{j}\}).
\end{align*}

Hence, in all three cases, the allocation $A$ is EF[1,1].
\end{proof}

\begin{lemma}
\label{lem:swap-stable}
The output allocation $A$ of Algorithm~\ref{alg:swap-stable-rr} is swap stable.
\end{lemma}
\begin{proof}
Let us consider a swap between participants $\mu^*_{q}$ and $\mu^*_{r}$, where $q,r\in Q$ with $q<r$.
Suppose that this swap is \emph{possibly} a beneficial swap, i.e., 
$v_{f(q)}(\mu^*_q)\le v_{f(q)}(\mu^*_r)$,
$v_{f(r)}(\mu^*_r)\le v_{f(r)}(\mu^*_q)$,
$f(q)\precsim_{\mu^*_q}f(r)$, and 
$f(r)\precsim_{\mu^*_r}f(q)$.
We will show that this swap cannot make any of the involved parties better off.
Denote by $\mu^{**}$ the matching that results from this swap.

If $v_{f(q)}(\mu^*_q) < v_{f(q)}(\mu^*_r)$, the matching $\mu$ can be improved by using $\mu^{**}$ instead, a contradiction.
So $v_{f(q)}(\mu^*_q) = v_{f(q)}(\mu^*_r)$.
Similarly, if $v_{f(r)}(\mu^*_r) < v_{f(r)}(\mu^*_q)$, then because $v_{f(q)}(\mu^*_q) = v_{f(q)}(\mu^*_r)$, the matching $\mu$ can again be improved by using $\mu^{**}$ instead, a contradiction.
So $v_{f(r)}(\mu^*_r) = v_{f(r)}(\mu^*_q)$.
Hence, the matching $\mu^{**}$ after the swap remains a feasible perfect matching in $G^*$.
As $\mu^*$ minimizes the sum of the participants' rank for teams among the perfect matchings in $G^*$, we get $f(q)\sim_{\mu^*_q}f(r)$ and $f(r)\sim_{\mu^*_r}f(q)$.
Therefore, the swap is not a beneficial swap, and the allocation $A$ is swap stable.
\end{proof}

\begin{lemma}
\label{lem:polynomial-time}
\Cref{alg:swap-stable-rr} can be implemented to run in polynomial time.
\end{lemma}

\begin{proof}
We first focus on computing the matching $\mu$ in \Cref{line:perfect-matching}.
The weight $w(1,\mu_1)$ can be found by simply taking the largest weight of an edge adjacent to vertex $1\in Q$ in $G$.
Given the weights $w(1,\mu_1), \dots, w(i-1,\mu_{i-1})$, to determine $w(i,\mu_i)$, we delete all edges $(q,p)$ with $1\le q \le i-1$ such that $w(q,p) \ne w(q,\mu_q)$ from $G$, change the weight of all edges $(q,p)$ with $i+1\le q\le m$ to $0$, and compute a maximum-weight perfect matching in the resulting graph.
Note that this matching can be found in time $O(m^3)$ \citep{Tomizawa71}.

Once we have $w(1,\mu_1),\dots,w(m,\mu_m)$, we can construct~$G^*$ in \Cref{line:matching-special} by keeping only the edges in $G$ such that $w(q,p) = w(q,\mu_q)$.
Finally, to compute $\mu^*$, we reassign the weight of each edge $(q,p)$ in $G^*$ to be the rank of participant~$p$ for team~$f(q)$ and find a minimum-weight perfect matching in $G^*$; again, this matching can be found in time $O(m^3)$.
\end{proof}

Next, we observe two ways in which \Cref{thm:balanced} cannot be improved:
the condition EF[1,1] cannot be strengthened to EF1,
and it is not possible to add individual stability to the list of guarantees.
In fact, the first observation was also made by \citet{ShoshanSeHa23}, although their work only deals with one-sided preferences.

\begin{proposition}
\label{prop:balanced-EF1}
Even for two teams with identical valuations, there does not necessarily exist a balanced EF1 allocation.
\end{proposition}

\begin{proof}
Consider an instance with $n = m = 2$ such that both teams have value $1$ for $p_1$ and $-1$ for $p_2$.
Clearly, no balanced allocation is EF1.
\end{proof}

\begin{proposition}
\label{prop:balanced-IS}
Even for two teams and nonnegative-value participants, there does not necessarily exist a balanced and individually stable allocation.
\end{proposition}

\begin{proof}
Consider an instance with $n = m = 2$ such that team~$1$ has value~$1$ for each participant, team~$2$ has value~$0$ for each participant, and both participants strictly prefer team~$1$ to team~$2$.
The only individually stable allocation assigns both participants to team~$1$, but this allocation is not balanced.
\end{proof}

\begin{algorithm*}[!tb]
  \caption{For computing an EF1, swap stable, and individually stable allocation}\label{alg:SS-EF1}
  Partition $P$ into $P^+ \coloneqq\{p\in P \mid \max_{i\in T}v_i(p)\ge 0\}$ and $P^- \coloneqq\{p\in P \mid \max_{i\in T}v_i(p)< 0\}$\;
  Let $\widehat{P}^+$ consist of $P^+$ together with $(n-1)|P^+|+n$ dummy participants, where each dummy participant yields value $0$ to every team and is indifferent between all teams\;
  Let $\widehat{P}^-$ consist of $P^-$ together with $(n-1)|P^-|+n$ dummy participants, where each dummy participant yields value $0$ to every team and is indifferent between all teams\;
  Let $A^+$ be the allocation obtained by executing \Cref{alg:swap-stable-rr} on $\widehat{P}^+$ with the teams in the forward order $1,2,\dots,n$\label{line:SS+}\;
  Let $A^-$ be the allocation obtained by executing \Cref{alg:swap-stable-rr} on $\widehat{P}^-$ with the teams in the backward order $n,n-1,\dots,1$\label{line:SS-}\;
  Return the allocation $A$ which is the union of $A^+$ and $A^-$ with the dummy participants removed\;
\end{algorithm*}

In spite of Propositions~\ref{prop:balanced-EF1} and \ref{prop:balanced-IS}, we show next that if we give up balancedness, we can attain EF1, swap stability, and individual stability simultaneously.
To this end, we combine our \Cref{alg:swap-stable-rr} with the \emph{double round-robin algorithm} introduced by \citet{AzizCaIg22}.
In the first phase, the participants who yield nonnegative value to at least one team are allocated by \Cref{alg:swap-stable-rr} in the forward order of the teams,
while in the second phase, the remaining participants are allocated by \Cref{alg:swap-stable-rr} in the backward order of the teams.
Intuitively, EF1 is guaranteed because, for each pair of teams $i$ and $j$ with $i<j$, $i$ does not envy $j$ in the first phase whereas $j$ does not envy $i$ in the second phase.
Moreover, we add a sufficient number of dummy participants, who yield value $0$ to every team and are indifferent between all teams, in order to guarantee individual stability.
This leads to each team receiving at least one dummy participant, and a beneficial deviation in the resulting situation can be captured by a beneficial swap between the deviating participant and a dummy participant.
The algorithm is formally described as \Cref{alg:SS-EF1}.

\begin{theorem}
\label{thm:individual-stable}
For any instance, \Cref{alg:SS-EF1} returns an EF1, swap stable, and individually stable allocation in polynomial time.
\end{theorem}

\begin{proof}
We show that \Cref{alg:SS-EF1} has the desired properties. 
Since \Cref{alg:swap-stable-rr} runs in polynomial time, so does \Cref{alg:SS-EF1}.
Note that $|\widehat{P}^+|=n(|P^+|+1)$ and $|\widehat{P}^-|=n(|P^-|+1)$, so every team has at least one dummy participant in each of $A^+$ and $A^-$.
Also, as \Cref{alg:swap-stable-rr} outputs a swap stable allocation, each of the allocations $A^+$ and $A^-$ is swap stable.
In addition, each participant $p\in P^+$ is allocated to a team that values her nonnegatively, i.e., $v_i(p)\ge 0$ for all $i\in T$ and $p\in A_i^+$.
Indeed, if $v_j(p)\ge 0>v_i(p)$ for some $i,j\in T$ and $p\in A_i^+$, the swap between $p$ and a dummy participant in~$A_j^+$ would lead to a better matching than $\mu$ in \Cref{alg:swap-stable-rr}, a contradiction.

We now prove that the allocation $A$ returned by \Cref{alg:SS-EF1} is EF1, swap stable, and individually stable.

\paragraph{EF1}
Consider any pair of teams $i$ and $j$ where $i<j$.

First, consider $i$'s envy for $j$.
In the first phase, $i$ has priority over $j$ and both teams receive the same number of participants, so $i$ does not envy $j$ with respect to $A^+$ (i.e., $v_i(A^+_i)\ge v_i(A^+_j)$).
Also, as $A^-$ is EF[1,1], there exist $X\subseteq A^-_i$ and $Y\subseteq A^-_j$ such that $|X|, |Y|\le 1$ and $v_i(A^-_i\setminus X)\ge v_i(A^-_j\setminus Y)$.
Hence, we obtain
\begin{align*}
    v_i(A_i\setminus X)
    &= v_i(A^+_i)+v_i(A^-_i\setminus X)
    \ge v_i(A^+_j)+v_i(A^-_j\setminus Y)
    \ge v_i(A^+_j)+v_i(A^-_j)
    =v_i(A_j);
\end{align*}
here, the second inequality holds because each participant in $Y$ yields negative value to every team.
Thus, $i$ does not envy~$j$ by more than one participant.

Next, consider $j$'s envy for $i$.
In the second phase, $j$ has priority over $i$ and both teams receive the same number of participants, so $j$ does not envy $i$ with respect to $A^-$ (i.e., $v_j(A^-_j)\ge v_j(A^-_i)$).
Also, as $A^+$ is EF[1,1], there exist $X'\subseteq A^+_j$ and $Y'\subseteq A^+_i$ such that $|X'|, |Y'|\le 1$ and $v_j(A^+_j\setminus X')\ge v_j(A^+_i\setminus Y')$.
Note that $v_j(X')\ge 0$ since $j$ only receives participants with nonnegative value in the first phase.
Hence, we obtain
\begin{align*}
    v_j(A_j)
    &=v_j(A^+_j)+v_j(A^-_j)
    \ge v_j(A^+_j\setminus X')+v_j(A^-_j)
    \ge v_j(A^+_i\setminus Y')+v_j(A^-_i)
    =v_j(A_i\setminus Y').
\end{align*}
Thus, $j$ does not envy $i$ by more than one participant.

\paragraph{Swap stability} 
Suppose to the contrary that the swap between some $p\in A_i$ and $q\in A_j$ is a beneficial swap in $A$.
Since each of $A^+$ and $A^-$ is swap stable, it cannot be that $p,q\in P^+$ or $p,q\in P^-$.
Thus, without loss of generality, we may assume that $p\in P^+$ and $q\in P^-$.
We have $v_i(p)\ge 0$ but $v_i(q)<0$, which means that the swap is not beneficial for team~$i$, a contradiction.

\paragraph{Individual stability} 
Suppose to the contrary that there is a beneficial deviation of participant~$p$ from team~$i$ to team~$j$ in $A$. 
If $p\in P^+$ (resp., $p\in P^-$), the swap between $p$ and a dummy participant in $A_j^+$ (resp., $A_j^-$) would be a beneficial swap in $A^+$ (resp., in $A^-$), contradicting the swap stability of $A^+$ (resp., $A^-$).
\end{proof}

\section{Pareto Optimality}
\label{sec:PO}

In this section, we turn our attention to Pareto optimality, which is a stronger requirement than both swap stability and individual stability.
Firstly, while it is easy to check whether an allocation is swap stable or individually stable by checking for all (polynomial number of) possible beneficial swaps or deviations, the same is not true for PO.

\begin{theorem}
\label{thm:PO-hardness}
Deciding whether an allocation is PO or not is coNP-complete, even for two teams with identical valuations, nonnegative-value participants, and a balanced allocation.
\end{theorem}

\begin{proof}
Checking that an allocation is Pareto dominated by another given allocation can be done in polynomial time, so the problem is in coNP.
To prove coNP-hardness, we reduce from \textsc{Subset Sum}.
An instance of \textsc{Subset Sum} consists of positive integers $b_1,\dots,b_r$ and $s$; it is a Yes-instance if and only if the sum of some subset of the $b_i$'s is exactly $s$.

Given an instance $(b_1,\dots,b_r;s)$ of \textsc{Subset Sum}, we create two teams with identical valuations for $2r$ participants; the values are $b_1,\dots,b_r,s,0,0,\dots,0$, where $0$ occurs $r-1$ times.
Participant~$p_{r+1}$ (with value $s$) prefers team~$1$ to team~$2$, while all other participants prefer team~$2$ to team~$1$.
Consider a balanced allocation in which the first $r$ participants are in team~$1$ while the other $r$ participants are in team~$2$.
We claim that this allocation admits a Pareto improvement if and only if $(b_1,\dots,b_r;s)$ is a Yes-instance.
Indeed, if $\sum_{i\in I}b_i = s$ for some $I\subseteq [r]$, then exchanging participants $p_i$  for $i\in I$ with $p_{r+1}$ yields a Pareto improvement: both teams are indifferent while all participants involved are better off.
For the converse direction, note that an exchange that yields a Pareto improvement cannot involve $p_{r+2},\dots,p_{2r}$, so such an exchange must involve $p_{r+1}$ along with a subset $P'\subseteq\{p_1,\dots,p_r\}$.
Since the teams have identical valuations, this exchange can be a Pareto improvement only when the $b_i$'s corresponding to $P'$ sum up to exactly~$s$.
\end{proof}

Note that even though the same decision problem is also coNP-complete for two teams with one-sided preferences \citep[Thm.~1]{AzizBiLa19}, it becomes trivial for any number of teams with identical valuations and one-sided preferences, because every allocation is PO in that case.
Also, the proof of \Cref{thm:PO-hardness} can be adapted to show a similar hardness for PO within the set of balanced allocations.
Indeed, we can use the same instance with the exception that all participants in team~$2$ prefer team~$1$ to team~$2$.
Then, there exists a Pareto improvement respecting the balancedness constraint if and only if the instance of \textsc{Subset Sum} is a Yes-instance.

In light of \Cref{thm:PO-hardness}, we cannot hope to reach a PO allocation in polynomial time by starting with an arbitrary allocation and iteratively finding Pareto improvements.
However, a PO allocation can be efficiently computed by simply assigning each participant to a team with the highest value for her, breaking ties in favor of a team that the participant prefers most.
Can we attain PO along with fairness for the teams?
The next example shows that round-robin-based algorithms such as \Cref{alg:swap-stable-rr,alg:SS-EF1} do not work, even for two teams with identical valuations and nonnegative-value participants.

\begin{example}
Consider the following instance with $n = 2$ and $m = 8$:
\begin{itemize}
\item $v_i(p_1) = v_i(p_2) = 4$, $v_i(p_3) = v_i(p_4) = 3$, $v_i(p_5) = v_i(p_6) = 2$, and $v_i(p_7) = v_i(p_8) = 1$ for $i\in \{1,2\}$;
\item $1\succ_{p_j}2$ for $j\in\{1,2,7,8\}$ and $2\succ_{p_j}1$ for $j\in\{3,4,5,6\}$.
\end{itemize}
Given this instance, \Cref{alg:swap-stable-rr,alg:SS-EF1} return an allocation $A$ that assigns to each team exactly one participant from each of the sets $\{p_1,p_2\}$, $\{p_3,p_4\}$, $\{p_5,p_6\}$, and $\{p_7,p_8\}$.
However, $A$ is Pareto dominated by the allocation $A' = (\{p_1,p_2,p_7,p_8\}, \{p_3,p_4,p_5,p_6\})$.
\end{example}

Nevertheless, 
for two teams and arbitrary-value participants, we can find an EF1 and PO allocation by extending the \emph{generalized adjusted winner procedure} of \citet{AzizCaIg22}.
The algorithm operates in a similar way as Aziz et al.'s algorithm, but we need to employ a tie-breaking rule among participants with the same ratio between the teams' values.

\begin{algorithm*}[!tb]
\caption{For computing an EF1 and PO allocation for two teams}\label{alg:gaw}
Assign each participant with zero value for both teams to a team that she prefers (breaking ties arbitrarily), and assume from now on that $P$ is the set of remaining participants\;
Let $P^*_1 = \{ p \in P \mid v_1(p) \geq 0, v_2(p)\leq 0 \}$ and $P^*_2 = \{ p \in P \mid v_1(p) \leq 0, v_2(p)\geq 0 \}$\;
Let $P^+ = \{ p \in P \mid v_1(p) > 0, v_2(p)> 0 \}$ and $P^- = \{ p \in P \mid v_1(p) < 0, v_2(p) < 0 \}$\;
Assume without loss of generality that the participants in $P^+\cup P^-$ are $p_1,p_2,\dots,p_r$, and relabel them so that $|v_{1}(p_1)|/|v_{2}(p_1)|\le |v_{1}(p_2)|/|v_{2}(p_2)|\le \dots \le |v_{1}(p_r)|/|v_{2}(p_r)|$.\label{line:tiebreak} 
Moreover, for participants with the same ratio, place them in the following order: \newline
(1) those in~$P^+$ who strictly prefer team~$2$ and those in~$P^-$ who strictly prefer team~$1$; \newline
(2) those in~$P^+ \cup P^-$ who like both teams equally; \newline
(3) those in~$P^+$ who strictly prefer team~$1$ and those in~$P^-$ who strictly prefer team~$2$\;
 Let $(A_1,A_2)\ot (P^+\cup P^*_1,P^-\cup P^*_2)$\label{line:initial-alloc}\; 
\For{$i\ot 1,2,\dots,r$}{
\lIf{team~$2$ does not envy team~$1$ by more than one participant}{\textbf{break}}
  \lIf{$p_i\in P^+$}{Move participant $p_i$ from team~$1$ to team~$2$ (i.e., $A_1\ot A_1\setminus \{p_i\}$ and $A_2\ot A_2 \cup \{p_i\}$)}
  \lElse{Move participant $p_i$ from team~$2$ to team~$1$ (i.e., $A_1\ot A_1\cup \{p_i\}$ and $A_2\ot A_2 \setminus \{p_i\}$)}
}
\Return $(A_1,A_2)$\;
\end{algorithm*}

\begin{theorem}
\label{thm:gaw}
Given any instance with two teams, there exists an algorithm that outputs an allocation that is EF1, PO, and team-PO in time $O(m^2)$.
\end{theorem}

\begin{proof}
The algorithm is shown as \Cref{alg:gaw}.
Since it is a version of the generalized adjusted winner procedure with specific tie-breaking, EF1, team-PO, and the running time follow from the work of \citet{AzizCaIg22}.
In particular, the tie-breaking in \Cref{line:tiebreak} takes time $O(m)$ and does not add to the overall running time.

It remains to show that the output allocation is PO.
First, participants with zero value for both teams are already assigned to a team that they prefer, and such participants do not affect the utility of either team no matter which team they are assigned to, so we may safely ignore them.
We claim that at any point from \Cref{line:initial-alloc} onward, the allocation~$A$ in the algorithm is PO.
Suppose to the contrary that there exists a Pareto improvement $A'=(A'_1,A'_2)$ of $A$.
Since this intermediate allocation~$A$ is team-PO \citep{AzizCaIg22}, we have
\begin{align*}
    v_1(A_1)=v_1(A'_1) \quad \text{and} \quad v_2(A_2)=v_2(A'_2).
\end{align*}
We can assume that participants in $P^*_1$ and $P^*_2$ stay in team~$1$ and $2$, respectively, because transferring such a participant makes neither team better off and at least one team worse off, contradicting team-PO.
Thus, in the following argument, we assume that only participants in $P^+\cup P^-$ are exchanged.
In addition, we observe from the proof of \citet{AzizCaIg22} that, for some value $\alpha$, 
$|v_1(p)|/|v_2(p)|=\alpha$ for all $p\in (A_1\cap A'_2)\cup (A_2\cap A'_1)$ (i.e., the exchanged participants between $A$ and $A'$).
We derive a contradiction by splitting the argument according to $p_i$, the last participant we moved in the for-loop (see \Cref{fig:adjusted_winner}). If we did not move any participant, we can apply the argument in Case~1.

\begin{figure*}[!t]
\begin{center}
  \begin{tikzpicture}[scale=0.8]
    \draw[thick,fill=blue!20] (-0.1,0)  rectangle (3.9,2);
    \draw[thick,fill=red!5] (4,0)  rectangle (16,2);
    \draw[thick,fill=blue!20] (16.1,0) rectangle (20.1,2);
    \draw (8,0) -- (8,2);
    \draw (12,0) -- (12,2);    
    \draw[above,blue!50!black] (1.9,2) node {$|v_1(p)|/|v_2(p)|<\alpha$};
    \draw[above, red!50!black] (10,2) node {$|v_1(p)|/|v_2(p)|=\alpha$};
    \draw[above,blue!50!black] (18.1,2) node {$|v_1(p)|/|v_2(p)|>\alpha$};
    \draw (6,1.5)  node[font=\small] {$p\in P^+\text{ and }2\succ_p 1$};
    \draw (6,0.5)  node[font=\small] {$p\in P^-\text{ and }1\succ_p 2$};
    \draw (10,1.5) node[font=\small] {$p\in P^+\text{ and }1\sim_p 2$};
    \draw (10,0.5) node[font=\small] {$p\in P^-\text{ and }1\sim_p 2$};
    \draw (14,1.5) node[font=\small] {$p\in P^+\text{ and }1\succ_p 2$};
    \draw (14,0.5) node[font=\small] {$p\in P^-\text{ and }2\succ_p 1$};
    \draw[<->,below,thick] (-0.1,-.3)  -- (7.95,-.3);
    \draw[<->,below,thick] (8.05,-.3)  -- (11.95,-.3);
    \draw[<->,below,thick] (12.05,-.3) -- (20.1,-.3);
    \node at (4,-.7)  {Case 1};
    \node at (10,-.7) {Case 2};
    \node at (16,-.7) {Case 3};    
  \end{tikzpicture}
  \caption{The order of participants in $P^+\cup P^-$ in the proof of \Cref{thm:gaw}.}\label{fig:adjusted_winner}
\end{center}
\end{figure*}

\paragraph{Case 1}
Suppose that one of the following holds: 
\begin{enumerate}[label=(\roman*)]
\item $|v_1(p_i)|/|v_2(p_i)| < \alpha$;
\item $|v_1(p_i)|/|v_2(p_i)| = \alpha$ and $p_i \in P^+$ with $2 \succ_{p_i} 1$; or
\item $|v_1(p_i)|/|v_2(p_i)| = \alpha$ and $p_i \in P^-$ with $1 \succ_{p_i} 2$. 
\end{enumerate} 
We see that every participant $p$ in $A_1 \cap P^-$ (resp.,~$A_2 \cap P^+$) with the ratio $|v_1(p)|/|v_2(p)|=\alpha$, if exists, strictly prefers team~$1$ (resp.,~team $2$), so such a participant cannot be exchanged for a Pareto improvement. 
Thus, $A_1\cap A'_2$ (resp., $A_2\cap A'_1$) consists only of participants in $P^+$ (resp., $P^-$).
Because $(A_1\cap A'_2)\cup (A_2\cap A'_1)$ is nonempty, the utility of team~$1$ is lower in $A'$ than in~$A$, which implies that $A'$ cannot be a Pareto improvement.

\paragraph{Case 2}
Suppose that $|v_1(p_i)|/|v_2(p_i)| = \alpha$, $p_i \in P^+\cup P^-$, and $1 \sim_{p_i} 2$.
Every participant $p$ in $A_1 \cap (P^+\cup P^-)$ with the ratio~$\alpha$ weakly prefers team~$1$, while 
every $p$ in $A_2 \cap (P^+\cup P^-)$ with the ratio $\alpha$ weakly prefers team~$2$. 
Note that any participant in $(A_1\cap A'_2) \cup (A_2\cap A'_1)$ likes both teams equally, because participants in $P^+\cup P^-$ with ratio~$\alpha$ and a strict preference are already allocated to their preferred team.
This together with the team-PO of $A$ implies that no participant is better off in $A'$ than in $A$. Thus, $A'$ is not a Pareto improvement.

\paragraph{Case 3}
This case is similar to Case~1.
Suppose that one of the following holds: 
\begin{enumerate}[label=(\roman*)]
\item $|v_1(p_i)|/|v_2(p_i)| > \alpha$;
\item $|v_1(p_i)|/|v_2(p_i)| = \alpha$ and $p_i \in P^+$ with $1 \succ_{p_i} 2$; or
\item $|v_1(p_i)|/|v_2(p_i)| = \alpha$ and $p_i \in P^-$ with $2 \succ_{p_i} 1$. 
\end{enumerate}
We see that every participant $p$ in $A_1 \cap P^+$ (resp., $A_2 \cap P^-$) with the ratio $\alpha$, if exists, strictly prefers team $1$ (resp., team~$2$), so such a participant cannot be exchanged for a Pareto improvement. 
Thus, $A_1\cap A'_2$ (resp., $A_2\cap A'_1$) consists only of participants in $P^-$ (resp., $P^+$).
Because $(A_1\cap A'_2)\cup (A_2\cap A'_1)$ is nonempty, the utility of team~$2$ is lower in $A'$ than in $A$, which implies that $A'$ cannot be a Pareto improvement.\\

In each of the three cases, we arrive at a contradiction. Therefore, we conclude that $A$ is PO.
\end{proof}

Although EF1, PO, and team-PO can be guaranteed simultaneously in the case of two teams, EF1 and participant-PO are already incompatible in this case.

\begin{proposition}
\label{prop:EF1-participant-PO}
Even for two teams with identical valuations and nonnegative-value participants, there does not necessarily exist an EF1 and participant-PO allocation.
\end{proposition}

\begin{proof}
Consider an instance with $n = 2$ and $m = 4$ such that each team has value $1$ for each participant and every participant prefers team~$1$ to team~$2$.
The only participant-PO allocation assigns all participants to team~$1$, but this allocation is not EF1.
\end{proof}

A similar counterexample holds for balanced allocations if we consider participant-PO within the set of balanced allocations.
Specifically, consider an instance with $n = 2$ and $m = 4$ such that each team has value~$0$ for each of two participants who prefer team~$1$, and value $1$ for each of two participants who prefer team~$2$.
The unique allocation that is participant-PO within the set of balanced allocations assigns the first two participants to team~$1$ and the last two participants to team~$2$, but this allocation is not EF1.

Note also that since PO is a stronger notion than individual stability, Proposition~\ref{prop:balanced-IS} implies that we cannot guarantee PO and balancedness simultaneously.
While one could ask whether there always exists a balanced allocation that is EF1 and PO within the set of balanced allocations, such a question remains unsolved even for one-sided preferences.\footnote{\citet{ShoshanSeHa23} explored a related question with category constraints but only addressed the case $n = 2$.}

We now move on to the general setting where the number of teams can be arbitrary.
Unfortunately, even for \emph{nonpositive-value} participants and one-sided preferences, it is unknown whether EF1 and PO can always be satisfied together \citep{EbadianPeSh22,GargMuQi22}.
We therefore restrict our attention to nonnegative-value participants in the remainder of this section.
By building upon a well-known result of \citet{CaragiannisKuMo19}, we can establish the existence of an EF1 and PO allocation.
For any allocation~$A$, its \emph{Nash welfare} is defined as the product $\prod_{i\in T}v_i(A_i)$.
An allocation is said to be a \emph{maximum Nash welfare (MNW) allocation} if it maximizes the Nash welfare among all allocations.\footnote{If the maximum possible Nash welfare is $0$, an MNW allocation should yield nonzero utility to the largest possible number of teams and, subject to that, maximize the product of utilities of these teams.}

\begin{theorem}\label{thm:MNW}
For any instance with nonnegative-value participants, there exists an EF1 and PO allocation.
\end{theorem}

\begin{proof}
Let $\mathcal{W}$ be the set of all MNW allocations, and let $A$ be an allocation that is PO within~$\mathcal{W}$---such an allocation must exist because otherwise there would be an infinite sequence of Pareto improvements in~$\mathcal{W}$.
It is known that every MNW allocation is EF1 \citep{CaragiannisKuMo19}, so $A$ is EF1.
We claim that $A$ is PO within the set of all allocations.
Suppose to the contrary that there is a Pareto improvement~$A'$ of~$A$.
Since $v_i(A_i') \ge v_i(A_i)$ for all $i\in T$, $A'$ must also be an MNW allocation.
However, this contradicts the assumption that $A$ is PO within $\mathcal{W}$.
\end{proof}

Given \Cref{thm:MNW}, a natural question is whether there exists a polynomial-time algorithm that computes an allocation guaranteed by the theorem.
However, this question is open even for one-sided preferences.\footnote{A pseudopolynomial-time algorithm for this problem  was given by \citet{BarmanKrVa18}.}
Next, we demonstrate that, in two special cases, such an algorithm exists.
The first case is when the teams have \emph{binary valuations}, meaning that each team has value either $0$ or $1$ for each participant.
In this case, it turns out that \Cref{alg:SS-EF1} already computes an EF1 and PO allocation in polynomial time.
(With binary valuations, the set $\widehat{P}^-$ in \Cref{alg:SS-EF1} is empty, so the algorithm can be simplified.)
Note that for one-sided preferences and binary valuations, an MNW allocation can be found in polynomial time \citep{DarmannSc15,BarmanKrVa18-2}, and such an allocation is guaranteed to be EF1 and PO \citep{CaragiannisKuMo19}.

\begin{theorem}
\label{thm:EF1-PO-binary}
For any instance with binary valuations, \Cref{alg:SS-EF1} computes an EF1 and PO allocation in polynomial time.
\end{theorem}

\begin{proof}
Since EF1 and polynomial-time computability were already shown in the proof of \Cref{thm:individual-stable}, it is sufficient to establish PO.

Let $A$ be the outcome of \Cref{alg:SS-EF1}, and suppose to the contrary that there is a Pareto improvement $A'$ of $A$.
For each participant~$p$, we denote by $A(p)$ and $A'(p)$ the team that $p$ is allocated to in $A$ and $A'$, respectively.
Note that $A'(p)\succsim_p A(p)$ for all $p\in P$ and $v_i(A'_i)\ge v_i(A_i)$ for all $i\in T$.
We claim that $v_{A(p)}(p)\ge v_{A'(p)}(p)$ for each participant $p$.
Indeed, if this is not the case, then  $v_{A(p)}(p)=0$ and $v_{A'(p)}(p)=1$ for some~$p$; we know that  $A'(p)\succsim_p A(p)$ for this~$p$.
However, a similar proof as that for individual stability in \Cref{thm:individual-stable} shows that such a deviation by $p$ from $A(p)$ to $A'(p)$, which hurts neither $p$ nor $A(p)$ and strictly helps $A'(p)$, cannot exist, thereby proving the claim.

Now, since $A'$ is a Pareto improvement of $A$, we have \[\sum_{p\in P}v_{A(p)}(p)=\sum_{i\in T} v_i(A_i)\le \sum_{i\in T}v_i(A'_i)=\sum_{p\in P}v_{A'(p)}(p).\]
Since $v_{A(p)}(p)\ge v_{A'(p)}(p)$ for all $p\in P$, we must have $v_{A(p)}(p) = v_{A'(p)}(p)$ for all $p\in P$ and $v_i(A_i) = v_i(A_i')$ for all $i\in T$.
Thus, we can construct a better matching than $\mu^*$ in \Cref{alg:swap-stable-rr} on $\widehat{P}^+$ (\Cref{line:SS+} of \Cref{alg:SS-EF1}) by a round-robin sequence in which each team $i$ picks participants in $A'_i$ as early as possible, because the Pareto improvement makes no participant worse off and at least one participant strictly better off.
However, this contradicts the definition of $\mu^*$.
\end{proof}

Next, we focus on the case of three teams with identical valuations and specific participant preferences.

\begin{algorithm*}[!htb]
\caption{For computing an EF1 and PO allocation for three teams with the conditions of \Cref{thm:three-teams}}\label{alg:three-teams}
\For{$i\ot 1,2,3$}{
Assign each participant with zero value of type $S_i$ to team $i$ (where $S_i$ denotes the type of participants who prefer team $i$)\;
} 
\While{there is at least one unassigned participant}{
Let $i$ be a team with the least current value. If there is more than one such team, choose a team~$i$ for which there is at least one unassigned participant of type~$S_i$, if possible\; \label{line:team-tiebreak}
\lIf{there is an unassigned participant of type~$S_i$}{Assign any such participant to team~$i$}
\lElseIf{there is only one type of participants left}{Assign any remaining participant to team~$i$}
\lElse{Denote the other two types by $S_j$ and $S_k$. Let $f(S_j)$ be the total value that team~$j$ would receive if all unassigned participants of type $S_j$ were assigned to it (in addition to the already assigned participants in team~$j$), and define $f(S_k)$ analogously for team $k$. Assign a participant of the type with the higher $f$-value to team $i$, breaking ties between types arbitrarily and breaking ties among participants in favor of higher-value participants\label{line:participant-tiebreak}}
}
\Return the current allocation $(A_1,A_2,A_3)$\;
\end{algorithm*}

\begin{theorem}
\label{thm:three-teams}
Suppose that there are $n = 3$ teams with identical valuations, all participants yield nonnegative value, and each participant prefers one team and is indifferent between the other two teams.
Then, there exists an algorithm that computes an EF1 and PO allocation in polynomial time.
\end{theorem}

\begin{proof}
The algorithm is shown as \Cref{alg:three-teams}, where for $i\in[3]$, we denote by $S_i$ the type of participants who prefer team~$i$.
It is clear that the algorithm runs in polynomial time.
Since the algorithm always assigns a participant to a team~$i$ with the least current value, no other team envies~$i$ by more than one participant at this point.
Hence, the same is true for all pairs of teams during the entire execution of the algorithm, which means that the returned allocation is EF1.\footnote{Alternatively, the algorithm can be seen as a special case of \citet{LiptonMaMo04}'s \emph{envy cycle elimination algorithm} for identical valuations.}

We now show that the allocation is PO.
Assume without loss of generality that the types run out in the order $S_1,S_2,S_3$.
In particular, team~$1$ may receive participants of all three types, team~$2$ may only receive participants of type $S_2$ and $S_3$, and team~$3$ may only receive participants of type~$S_3$.
Suppose for contradiction that there exists a Pareto improvement $A' = (A_1',A_2',A_3')$ of the output allocation~$A$.
Denoting the common team valuation by $v$, we have $v(A_i) = v(A_i')$ for all $i\in\{1,2,3\}$.
Since all participants with zero value are already with their preferred team in $A$, they must remain with their team in $A'$.
Moreover, since $A_3$ only contains participants of type~$S_3$ and $v(A_3) = v(A_3')$, it must be that $A_3 = A_3'$.
So, from $A$ to $A'$, some participants of type $S_3$ are moved from $A_2$ to $A_1'$, while some participants of type $S_2$ or $S_3$ (at least one participant of type $S_2)$ are moved from $A_1$ to $A_2'$, where both sets of participants have the same total value.
We will show that every participant of type~$S_2$ in~$A_1$ has a strictly larger value than the total value of all participants of type~$S_3$ in $A_2$.
This is sufficient to obtain the desired contradiction.

Consider the moment when the algorithm assigns the last participant $p$ of type~$S_2$ to team~$1$.
Since participants of type~$S_3$ run out after those of type~$S_2$, there is at least one participant of type~$S_3$ available at this moment.
The choice of the algorithm to assign a participant of type~$S_2$ to team~$1$ implies that $f(S_2) \ge f(S_3)$.
After $p$'s assignment, $f(S_2)$ decreases by $v(p)$, so it holds that $f(S_3) - f(S_2) \le v(p)$ at this point.
Now, because participants of type~$S_2$ have not run out before $p$'s assignment, the first assignment of a participant $\widehat{p}$ of type~$S_3$ to team~$2$ must occur after $p$'s assignment.
Between $p$'s assignment and $\widehat{p}$'s assignment, some participants of type~$S_3$ may be assigned to team~$1$---this only decreases $f(S_3)$.
Hence, directly before $\widehat{p}$'s assignment, we still have $f(S_3) - f(S_2) \le v(p)$. 
Moreover, at this point, the partial allocation $A''$ satisfies $v(A''_2) < v(A''_3)$ (if $v(A''_2) = v(A''_3)$, the algorithm should have assigned $\widehat{p}$ to team~$3$ due to the tie-breaking rule in \Cref{line:team-tiebreak}), and only participants of type~$S_3$ are left, i.e., $v(A''_2) = f(S_2)$.
Therefore, the total value of participants of type~$S_3$ assigned to team~$2$ is at most $f(S_3) - v(A''_3) < f(S_3) - v(A''_2) =  f(S_3) - f(S_2) \le v(p)$, i.e., this value is strictly less than $v(p)$.
On the other hand, by the tie-breaking on participants (\Cref{line:participant-tiebreak}), every participant of type~$S_2$ assigned to team~$1$ has value at least $v(p)$.
This completes the proof.
\end{proof}

Finally, we provide a pseudopolynomial-time algorithm for the case where the number of teams is constant.

\begin{theorem}
\label{thm:const}
For any instance with a constant number of teams, each of which has a nonnegative integer value for each participant,
an EF1 and PO allocation can be computed in pseudopolynomial time.
\end{theorem}

\begin{proof}
Let $v_{\max}\coloneqq\max_{i\in T,\,p\in P}v_i(p)$.
We construct a table $H$ which classifies all possible utility vectors for teams that can be attained by allocating the first $j$ participants $p_1,\dots,p_j$.
The entry $H(\bm{u},j)$ indicates whether there exists an allocation $A$ of participants $p_1,\dots,p_j$ such that $\bm{u}=(v_1(A_1),\dots,v_n(A_n))$.
Moreover, if there exists such an allocation, $H(\bm{u},j)$ is an allocation that maximizes the participants' happiness lexicographically with respect to the reverse participant order (i.e., maximizes participant $p_j$'s happiness, then maximizes participant $p_{j-1}$'s happiness, and so on) among such allocations.
Note that the utility of a team for an allocation is an integer belonging to the range $[0,\, m\cdot v_{\max}]$. Hence, the size of the table $H$ is $O(m\cdot (1+m\cdot v_{\max})^n)$, which is pseudopolynomial when $n$ is a constant.
We can fill in the entries of the table according to the following recursive formula, where $\chi_i$ denotes the $i$th unit vector of length $n$, that is, the $k$th coordinate is $1$ if $k=i$ and $0$ otherwise.
\begin{itemize}
    \item For $j=0$, the entry $H(\bm{u},j)$ is $(\emptyset,\dots,\emptyset)$ if $\bm{u}=(0,\dots,0)$, and $\bot$ otherwise.\footnote{The symbol $\bot$ stands for a ``null'' value.}
    \item For $j=1,2,\dots,m$, the entry $H(\bm{u},j)$ is $\bot$ if $H(\bm{u}-v_i(p_j)\cdot\chi_{i},\,j-1)=\bot$ for all $i\in T$.
    Otherwise, let $i^*$ be a team that $p_j$ prefers the most among the teams $i$ such that $H(\bm{u}-v_i(p_j)\cdot\chi_{i},\,j-1)\ne\bot$. If there are multiple such teams, we select a team that yields a lexicographically optimal allocation for $p_1,\dots,p_{j-1}$ with respect to the reverse participant order. Then, the entry is the allocation such that $p_1,\dots,p_{j-1}$ are allocated as in $H(\bm{u}-v_{i^*}(p_j)\cdot\chi_{i^*},\,j-1)$ while $p_j$ is allocated to $i^*$.
\end{itemize}
The entries $H(\bm{u},j)$ can be computed in $O(nm)$ time each in a bottom-up manner, so we can construct the table $H$ in $O(nm^2\cdot (1+m\cdot v_{\max})^n)$ time.
Now, by using the table, we can pick a utility vector $\bm{u}^*$ that corresponds to an MNW allocation (of all $m$ participants).
Similarly to the proof of \Cref{thm:MNW}, we can then conclude that the allocation $H(\bm{u}^*,m)$ is EF1 and PO. 
\end{proof}

\section{Justified Envy-Freeness}
\label{sec:justifiedEF}

In this section, we consider justified EF.
Note that if $m = n$ and all participants yield nonnegative value to every team, the existence of an EF1 and justified EF allocation follows from the celebrated result in two-sided matching of \citet{GaleSh62} (with arbitrary tie-breaking).
We show that, perhaps surprisingly, this guarantee cannot be extended to the case $m > n$.

\begin{proposition}
\label{prop:justifiedEF-counterexample}
Even for two teams and nonnegative-value participants who prefer the same team, there does not necessarily exist an EF1 and justified EF allocation.
\end{proposition}

\begin{proof}
Consider an instance with $n = 2$ and $m = 4$ such that team~$1$ has value $3,3,2,2$ for the four participants, respectively, team~$2$ has value $1,1,0,0$, respectively, and all participants strictly prefer team~$1$ to team~$2$.
Team~$2$ needs at least one of $p_1$ and $p_2$ in order for EF1 to be fulfilled---assume without loss of generality that it receives $p_2$.
Given this, team~$1$ needs at least one of $p_3$ and $p_4$ in order for EF1 to be fulfilled---assume without loss of generality that it receives $p_3$.
However, this results in $p_2$ having justified envy toward $p_3$.
\end{proof}

In fact, the incompatibility in Proposition~\ref{prop:justifiedEF-counterexample} persists even if EF1 is weakened to \emph{envy-freeness up to $k$ participants (EF$k$)} for any fixed~$k$, where $k$ participants may be removed (rather than just one participant) in order to eliminate the envy.
To see this, consider a similar instance with $n = 2$ and $m = 4k$ such that
\begin{itemize}
\item $v_1(p_j) = 3$ and $v_2(p_j) = 1$ for $j\in[k+1]$;
\item $v_1(p_j) = 2$ and $v_2(p_j) = 0$ for $j\in\{k+2,\dots,4k\}$;
\item $1\succ_{p_j} 2$ for $j\in[4k]$.
\end{itemize}
Team~$2$ needs at least one of $p_1,\dots,p_{k+1}$ for EF$k$ to be satisfied. 
Given this, in order to avoid justified envy, team~$2$ must receive all of $p_{k+2},\dots,p_{4k}$ as well.
However, this results in EF$k$ from team~$1$ toward team~$2$ being violated.

Can we efficiently determine whether an EF1 and justified EF allocation exists in a given instance?
The following theorem gives a negative answer, provided that P $\ne$ NP.

\begin{table*}[!tb]
	\footnotesize

	\centering
	\begin{tabular}{llcccc}
		\toprule
        & & participant $a$ & participant $b$ & participant $c_j$ $(c_j \in C)$ & participant $p_w$ $(w \in V)$\\ 
        \midrule
        \multicolumn{2}{l}{Preferences} & $s \succ d \succ \cdots $ & $s \succ d \succ \cdots $ &  arbitrary & $s \succ t_{\phi(w)} \succ \cdots $\\  
        \midrule
        &team $s$ & $1$ & $1$ & $1- \varepsilon$ &  $1$\\  
        Valuations&team $t_e$ ($e \in E$) & $0$ & $0$ & $0$ &  $1$ or $0$ (see caption)\\  
        &team $d$ & $1$ & $1$ & $\varepsilon$ &  $\varepsilon$\\  
		\bottomrule
	\end{tabular}

\caption{Participant preferences and team valuations in the proof of \Cref{thm:justifiedEF-hardness}.
Team~$t_e$ has value~$1$ for participant~$p_w$ if $w$ is an endpoint of $e$, and $0$ otherwise.
}
\label{table:team:valuation}
\end{table*}

\begin{theorem}
\label{thm:justifiedEF-hardness}
Even for nonnegative-value participants with strict preferences, deciding whether there exists an EF1 and justified EF allocation is NP-complete.
\end{theorem}

\begin{proof}
The problem belongs to NP since it can be verified in polynomial time whether a given allocation is EF1 and justified EF. 
We prove the NP-hardness by reducing from {\sc Independent Set}, the NP-complete problem of deciding whether a graph $G$ admits an independent set of size~$k$, i.e., a set of $k$ vertices no two of which are connected by an edge \citep[p.~194]{GareyJo79}.
Consider an instance $(G,k)$ of {\sc Independent Set}, where $G=(V,E)$. 
Without loss of generality,\footnote{If $k = 1$ or $|V| \le 3$, the problem is trivial. If $|V| > |E|$, we repeatedly remove an isolated vertex and decrease $k$ by $1$ each time until $|V| \le 2|E|$, then add an isolated clique of size $|V|$ and increase $k$ by $1$.} assume that $k \geq 2$ and $|E| \geq |V| \ge 4$. 
We construct an instance of our problem as follows. 

We set $\tau \coloneqq 2+|E|$, which will be the number of teams in our instance, and $\eta \coloneqq |V|+(k+1)(\tau-1)+3$, which will be the number of participants in our instance. Clearly,  $\eta>\tau$ since $(k+1)(\tau-1)>\tau$ for $\tau \geq 2$.
We set $\varepsilon \in (0,1)$ such that $\varepsilon < \min\{1/(\eta-2),1-k/(k+1)\}$; hence, we have $1 > (\eta -2)\varepsilon$ and $(k+1)(1-\varepsilon)>k$. 

Create one special team $s$. 
For each edge $e \in E$, create an edge team $t_e$. 
For each vertex $w \in V$, create a vertex participant $p_w$.  
The special team $s$ assigns value $1$ to each vertex participant $p_w$ for $w \in V$. 
Each edge team $t_e$ assigns value $1$ to a vertex participant $p_w$ if $w$ is an endpoint of edge $e$ and value~$0$ otherwise. 
Create an injective map $\phi \colon V \rightarrow E$; this is possible since $|V|\le |E|$.
Each vertex participant~$p_w$ prefers $s$ the most and $t_{\phi(w)}$ the second most, and has an arbitrary preference over the other teams (including ones that will be defined later). 

For the special team $s$, we create the following gadget $I$ that forces $s$ to get $k$ vertex participants in our desired allocation. The gadget $I$ consists of two teams $s$ and $d$, two participants $a$ and $b$, and $(k+1)(\tau -1)+1$ participants $c_{j}$ for $j \in [(k+1)(\tau -1)+1]$. Let $C=\{c_1,c_2,\ldots,c_{(k+1)(\tau -1)+1}\}$. Participants $a$ and $b$ prefer $s$ the most and $d$ the second most, and have an arbitrary preference over the other teams. Each $c_j \in C$ has an arbitrary (strict) preference over the teams. 
Teams $s$ and $d$ are the only teams that have positive value for the participants in the gadget $I$. 
Team $s$ assigns value $1$ to each of $a$ and $b$ and value $1-\varepsilon$ to each participant $c_j \in C$. 
Team $d$ assigns value $1$ to each of $a$ and $b$, and value $\varepsilon$ to the remaining participants (including those outside the gadget). 
See Table~\ref{table:team:valuation} for a summary of the participants' preferences and teams' valuations in our constructed instance. 
This completes the description of our construction, which clearly runs in polynomial time. 

To finish the proof, we show that there exists an allocation~$A$ satisfying EF1 and justified EF if and only if there exists an independent set $S \subseteq V$ of size $k$ in $G$.

\vspace{2mm}

($\Rightarrow$)
First, we show that if there exists an allocation~$A$ satisfying EF1 and justified EF, then there exists an independent set $S \subseteq V$ of size $k$ in $G$. 
Suppose that there exists such an allocation~$A$. 
We first claim that in $A$, team $s$ must receive at least $k$ vertex participants. 
To see this, consider the gadget $I$. 
Observe that $d$ must be allocated at least one of the participants $a$ and $b$. 
Indeed, if $d$ receives no participant, then since $\eta >\tau $, there is a team that is allocated at least two participants positively valued by $d$, contradicting the EF1 property of $A$.
Hence, $d$ is allocated at least one participant. 
Moreover, if $d$ is allocated neither $a$ nor $b$, then these two participants must be allocated to team $s$, since otherwise either $a$ or $b$ would have justified envy toward a participant allocated to team $d$; however, this would violate the EF1 condition from the viewpoint of $d$ since the maximum value $d$ can achieve from participants other than $a$ and $b$ is $(\eta -2)\varepsilon$, which is strictly less than $1$. 
Thus, at least one of $a$ and $b$ must be allocated to team $d$; without loss of generality, assume that $b$ is allocated to $d$. 
This means that none of the participants $c_j \in C$ can be allocated to team $s$, since otherwise $b$ would have justified envy toward such a participant. 
Hence, at most $\tau -1$ teams receive a participant from $C$; since $|C| = (k+1)(\tau-1)+1$, at least one team $t \neq s$ gets at least $k+2$ such participants. 
Thus, the value that team $s$ can obtain from the gadget $I$ is at most $1$, and there is a team $t$ that receives a subset of participants whose total value is at least $(k+2)(1-\varepsilon)$ from the viewpoint of team $s$. Since $A$ is EF1, team $s$ must receive value at least $(k+1)(1-\varepsilon)-1$, which is strictly greater than $k-1$, from outside of the gadget $I$.
The only way this is possible is for team $s$ to receive at least $k$ vertex participants. 

Now, let $S'$ denote the set of vertices in $V$ that correspond to the vertex participants allocated to team $s$. 
We have seen that $S'$ contains at least $k$ vertices. 
We claim that $S'$ forms an independent set in $G$. 
To see this, suppose for contradiction that $S'$ contains a pair of vertices that form an edge $e$ in $G$. 
Then, the edge team $t_e$ does not receive the participants positively valued by itself and therefore envies team $s$ by more than one participant, a contradiction. 
Thus, $S'$ is an independent set of size at least $k$, and any subset $S\subseteq S'$ of size $k$ forms an independent set as well.

\vspace{2mm}

($\Leftarrow$)
Conversely, suppose that there exists an independent set $S$ of size $k$ in $G$. 
Take an edge $e^*$ such that $\phi(w)=e^*$ for some $w \in S$; note that since $\phi$ is injective, we have $e^* \neq \phi(w)$ for all $w \in V \setminus S$. 
Consider the following allocation $A$. 
First, it allocates all the vertex participants $p_w$ with $w \in S$ to $s$ and each of the remaining $|V|-k$ vertex participants $p_w$ to the edge team $t_{\phi(w)}$; then, it allocates participant~$a$ to team $s$ and participant~$b$ to team $d$; finally, it allocates the participants $c_j \in C$ to the teams other than $s$ in such a way that the edge team $t_{e^*}$ receives $k+2$ such participants and every other team receives $k+1$ such participants. 

We first show that $A$ is EF1. 
\begin{itemize}
    \item Consider team $s$. Team $s$ gets value $k+1$. Team $d$ gets value $1+(k+1)(1-\varepsilon)$ from the viewpoint of $s$ and the envy of $s$ would disappear by removing participant $b$. 
    Every edge team $t_e$ with $\phi(w)=e$ for some $w \in V \setminus S$ gets value at most $1+(k+1)(1-\varepsilon)$ from the viewpoint of $s$ and the envy of~$s$ would disappear upon removing the vertex participant $p_w$ allocated to $t_e$. 
    The edge team $t_{e^*}$ gets value $(k+2)(1-\varepsilon)$ from the viewpoint of $s$ and the envy of~$s$ would disappear upon removing one of the participants allocated to $t_{e^*}$.  
    Every other edge team $t_e$ with $e \neq e^*$ and $e \neq \phi(w)$ for all $w \in V \setminus S$ obtains value at most $(k+1)(1-\varepsilon)$ from the viewpoint of $s$ and therefore $s$ does not envy $t_e$. 
    \item Consider an arbitrary edge team $t_e$. Team $t_e$ does not envy $d$ since $d$ does not receive a vertex participant. Team $t_e$ does not envy $s$ by more than one participant since $s$ is allocated at most one vertex participant $p_w$ corresponding to an endpoint of $e$. 
    Team $t_e$ does not envy edge team $t_{e'}$ with $e' \neq \phi(w)$ for all $w \in V \setminus S$ since $t_e$ assigns value $0$ to participants in $C$. 
    Team $t_e$ may envy edge team $t_{e'}$ with $e' = \phi(w)$ for some $w \in V \setminus S$ but the envy can be eliminated by removing the vertex participant allocated to $t_{e'}$. 
    
    \item Consider team $d$. Team $d$ does not envy the other teams by more than one participant, since $d$ receives value $1+(k+1)\varepsilon$ and every team other than $d$ gets value at most $1$ from the viewpoint of $d$ after removing some participant assigned to the team (recall that $1>(\eta-2)\varepsilon \geq (k+2) \varepsilon$). 
\end{itemize}

Now, it remains to show that $A$ satisfies justified EF. 
\begin{itemize}
    \item Consider an arbitrary vertex participant $p_w$. 
    Each participant $p_w$ with $w \in S$ does not have justified envy toward the other participants since she is allocated to her first choice. 
    Each participant $p_w$ with $w \not \in S$ is allocated to the edge team $t_{\phi(w)}$ and she may envy the participants allocated to $s$; however, these participants are valued at $1=v_s(p_w)$ by $s$, which means the envy is not justified. 
    \item Consider participant $a$. 
    Participant $a$ does not envy the other participants since she is allocated to her first choice.
    \item Consider participant $b$. 
    Participant $b$ may envy the participants allocated to $s$, but these participants are valued at $1=v_s(b)$ by $s$, which means the envy is not justified. 
    \item Consider participant $c_j \in C$. 
    Since every team values $c_j$ no more than any other participant, any envy that $c_j$ has toward another participant is not justified.
\end{itemize}
This completes the proof.
\end{proof}

Despite \Cref{thm:justifiedEF-hardness}, we show next that the problem becomes efficiently solvable if there are two teams.
Note that this special case covers the example in the proof of Proposition~\ref{prop:justifiedEF-counterexample}.

\begin{algorithm*}[!tb]
  \caption{For deciding whether an EF1 and justified EF allocation exists under the conditions of \Cref{thm:justifiedEF-algo-two} }\label{alg:justifiedEF-two}
  Let $P_1$ be the set of participants who prefer team~$1$, and $P_2 = P\setminus P_1$; \quad \emph{// Participants in $P_2$ prefer team~$2$}\\
  Let $V_1 = \{v_1(p)\mid p\in P\}$ and $V_2 = \{v_2(p)\mid p\in P\}$\;
  Let $W_1 = \{(x_1,n_1)\mid \text{either } x_1\in V_1 \text{ and } n_1\in [m], \text{ or } x_1 = -\infty \text{ and } n_1 = 0 \}$; \quad \emph{// $x_1$ represents the maximum value for team~$1$ of a participant in $P_1$ assigned to team~$2$, and $n_1$ represents the number of participants attaining this maximum value}\\
  Let $W_2 = \{(x_2,n_2)\mid \text{either } x_2\in V_2 \text{ and } n_2\in [m], \text{ or } x_2 = -\infty \text{ and } n_2 = 0 \}$; \quad \emph{// $x_2$ represents the maximum value for team~$2$ of a participant in $P_2$ assigned to team~$1$, and $n_2$ represents the number of participants attaining this maximum value}\\
  \For{$(x_1,n_1)\in W_1$ and $(x_2,n_2)\in W_2$}{
  Assign all participants $p\in P$ such that $v_1(p) < x_1$ to team~$2$, and all participants $p\in P$ such that $v_2(p) < x_2$ to team~$1$\; \label{line:assign-1}
  Assign all participants $p\in P_1$ such that $v_1(p) > x_1$ to team~$1$, and all participants $p\in P_2$ such that $v_2(p) > x_2$ to team~$2$\; \label{line:assign-2}
  \lIf{there is a conflict in the preceding assignments (i.e., some participant has been assigned to both teams)}{\textbf{continue}}
  Let $Q_1$ be the participants from $P_1$ who are still unassigned; \quad \emph{// $v_1(p) = x_1$ for every $p\in Q_1$}\\
  Let $Q_2$ be the participants from $P_2$ who are still unassigned; \quad \emph{// $v_2(p) = x_2$ for every $p\in Q_2$}\\
  Among the participants in $Q_1$, assign to team~$2$ the ones with the highest value for team~$2$ (breaking ties arbitrarily) so that there are a total of $n_1$ participants $p\in P_1$ with $v_1(p) = x_1$ assigned to team~$2$\;
  \lIf{the preceding assignment is not possible}{\textbf{continue}}
  Assign all remaining participants in $Q_1$ to team~$1$\;
  Among the participants in $Q_2$, assign to team~$1$ the ones with the highest value for team~$1$ (breaking ties arbitrarily) so that there are a total of $n_2$ participants $p\in P_2$ with $v_2(p) = x_2$ assigned to team~$1$\;
  \lIf{the preceding assignment is not possible}{\textbf{continue}}
  Assign all remaining participants in $Q_2$ to team~$2$\;
  \lIf{the resulting allocation is EF1}{\Return this allocation}
  } 
  \Return None\;
\end{algorithm*}

\begin{theorem}
\label{thm:justifiedEF-algo-two}
For two teams and nonnegative-value participants with strict preferences, there is a polynomial-time algorithm that decides whether an EF1 and justified EF allocation exists (and, if so, computes such an allocation).
\end{theorem}

\begin{proof}
Let $P_1$ be the set of participants who prefer team~$1$, and define $P_2$ analogously.
Given an allocation, we let $x_1$ be the maximum value for team~$1$ of a participant in $P_1$ who is assigned to team~$2$, and let $n_1$ be the number of participants attaining this maximum value.
If there is no such participant, we let $x_1 = -\infty$ and $n_1 = 0$.
Define $x_2$ and $n_2$ analogously.
The idea behind the algorithm is that, once we fix $x_1,n_1,x_2,n_2$, we can efficiently check whether there is an EF1 and justified EF allocation consistent with these values.
This allows us to iterate over all possible values of these parameters.
The pseudocode of the algorithm is given as \Cref{alg:justifiedEF-two}.

To see that the algorithm is correct, consider an EF1 and justified EF allocation with the associated values $x_1,n_1,x_2,n_2$.
If a participant $p\in P$ with $v_1(p) < x_1$ is assigned to team~$1$, then a participant $p'\in P_1$ with $v_1(p') = x_1$ assigned to team~$2$ would have justified envy toward~$p$.
Hence, any participant $p\in P$ with $v_1(p) < x_1$ must be assigned to team~$2$; this is also vacuously true if $x_1 = -\infty$.
Moreover, by definition of $x_1$, any participant $p\in P_1$ with $v_1(p) > x_1$ must be assigned to team~$1$.
Analogous arguments apply to $v_2$.
It follows that the relevant participants must be assigned as in \Cref{line:assign-1,line:assign-2} of \Cref{alg:justifiedEF-two}.
Moreover, once these assignments are made, justified EF is guaranteed.

At this point, the participants who may still be unassigned are participants $p\in P_1$ with $v_1(p) = x_1$ and participants $p\in P_2$ with $v_2(p) = x_2$.
Let $Q_1$ be the set of participants $p\in P_1$ with $v_1(p) = x_1$ who are still unassigned, and define $Q_2$ analogously.
Since team~$1$ is indifferent between participants in $Q_1$, in order to check whether EF1 can be satisfied, it suffices to assign to team~$2$ a subset of these participants with the highest value for team~$2$, with the size of the subset chosen so that the definition of $n_1$ is fulfilled.
A similar statement holds for $Q_2$.
The correctness of the algorithm then follows from the fact that it checks all possible values of $x_1,n_1,x_2,n_2$.

The number of possible values of $(x_1,n_1,x_2,n_2)$ is $O(m^4)$, and for each $(x_1,n_1,x_2,n_2)$ the algorithm takes time $O(m)$.
Hence, \Cref{alg:justifiedEF-two} runs in time $O(m^5)$.
\end{proof}

Finally, we prove that if the two teams have identical valuations, then an EF1 and justified EF allocation always exists.

\begin{theorem}
\label{thm:justifiedEF-algo-two-iden}
For two teams with identical valuations and nonnegative-value participants, there exists an EF1 and justified EF allocation, and such an allocation can be computed in polynomial time.
\end{theorem}

\begin{proof}
To show the theorem, we will align the participants on a path so that the values form a ``valley'', and apply a discrete ``cut-and-choose'' algorithm to construct an EF1 and justified EF allocation. 

Formally, given a valuation function $v$ and a path $\mathcal{P}=(p_1,p_2,\ldots,p_t)$ of participants, $p_j$ is called a \emph{lumpy tie} if 
\[
v(\{p_1,\ldots,p_j\}) \geq v(\{p_{j+1},\ldots,p_t\})
\]
and 
\[
v(\{p_1,\ldots,p_{j-1}\}) \leq v(\{p_{j},\ldots,p_t\}).
\]
Such a participant exists, e.g., by taking a participant $p_j$ with the smallest index $j$ such that $v(\{p_1,\ldots,p_j\}) \geq v(\{p_{j+1},\ldots,p_t\})$. 
Using the notion of a lumpy tie, \citet[Def.~3.1]{BiloCaFl22} developed the following cut-and-choose algorithm over a path that computes an EF1 allocation for two teams with identical valuations in polynomial time. 

\begin{algorithm*}[!h]
\caption{Discrete cut-and-choose over a path}
\label{alg:twoteams}
\SetInd{0.8em}{0.3em}
{\bf Input} $\mathcal{P}=(p_1,p_2,\ldots,p_t)$ and valuation function $v$\;
Let $p_j$ denote the leftmost lumpy tie of the path $P$\;
Let $L=\{p_1,\ldots,p_{j-1}\}$ and $R=\{p_{j+1},\ldots,p_t\}$\;
\If{$v(L) \geq v(R)$}
{\Return $A=(L,R \cup \{p_j\})$;}
\Else{
\Return $A=(L \cup \{p_j\},R)$;
}
\end{algorithm*}

Now, consider an instance consisting of two teams with the same nonnegative valuation $v$. 
For $i \in [2]$, denote by $S_i$ the set of participants who prefer team $i$ over the other team (ties broken arbitrarily). 
Relabel the participants in $S_1=\{y_1,y_2,\ldots,y_{\ell}\}$ so that $v(y_1) \leq v(y_2) \leq \cdots \leq v(y_{\ell})$, and relabel the participants in $S_2=\{z_1,z_2,\ldots,z_{h}\}$ so that $v(z_1) \leq v(z_2) \leq \cdots \leq v(z_h)$. 
Consider the path $\mathcal{P}^* \coloneqq (y_{\ell},\ldots,y_{1},z_1,\ldots,z_h)$.

Apply Algorithm~\ref{alg:twoteams} to $\mathcal{P}^*$ and $v$. 
Let $A$ denote the resulting allocation. 
The allocation $A$ is EF1 \citep[Prop.~3.2]{BiloCaFl22}. 
To show that $A$ is justified EF, assume first that the leftmost lumpy tie over $\mathcal{P}^*$ is a participant $y_i$ for some $i \in [\ell]$. 
Every participant $z_k$ for $k\in [h]$, as well as every participant $y_j$ with $j > i$, is already allocated to her favorite team.
Consider a participant $y_j$ with $j \leq i$ who is allocated to team~$2$. The participant $y_j$ may envy a participant $y_{j'}$ with $j' \geq i$ who is allocated to team $1$; however, since $v(y_{j'}) \geq v(y_j)$, this envy is not justified. 
The proof proceeds similarly if the leftmost lumpy tie over $\mathcal{P}^*$ is a participant $z_i$ for some $i \in [h]$. 
\end{proof}

\section{Conclusion}

In this work,
we have investigated the setting of fair division with two-sided preferences, which serves to model the allocation of participating players to sports teams, employees to branches of a restaurant chain, or volunteers to community service clubs.
Our focus is on fairness among the teams together with stability for both sides.
We showed that EF1, swap stability, and individual stability are compatible in this setting, and an allocation satisfying these properties can be computed in polynomial time even when the teams may have positive or negative values for the participants.
If all participants yield nonnegative value to the teams, an EF1 and PO allocation always exists, and such an allocation can be found efficiently provided that the values are binary.
Furthermore, we demonstrated that an EF1 and justified EF allocation does not always exist and determining whether such an allocation exists is NP-complete.

For future work, it would be worth examining the interplay between other common (one-sided) fairness notions and two-sided stability conditions such as swap stability and justified EF.
For instance, while we have concentrated on the important fairness notion of EF1, 
one might try to achieve an approximation of \emph{maximin share fairness (MMS)} \citep{Budish11,KurokawaPrWa18}, either in place of or in conjunction with EF1.
Note that EF1 implies a $1/n$-approximation of MMS for nonnegative-value participants \citep{AmanatidisBiMa18}, so our relevant results also hold for this approximation.
It could also be interesting to extend our results to accommodate teams with varying \emph{entitlements} \citep{FarhadiGhHa19,BabaioffEzFe21,BabaioffNiTa21,ChakrabortyIgSu21}; this would allow us to capture restaurant branches or community service clubs of different sizes.\footnote{For nonnegative-value participants and one-sided preferences, weighted generalizations of EF1 can be attained by \emph{picking sequences} (of which round-robin is a special case) \citep{ChakrabortyScSu21,ChakrabortySeSu22}.
As a result, we can achieve a weighted extension of EF1 together with swap stability by applying Algorithm~\ref{alg:swap-stable-rr} to the desired picking sequence instead of the round-robin sequence.}
Finally, one could attempt to bring other concepts from the rich matching literature into consideration as well.

\section*{Acknowledgments}

This work was partially supported by JSPS KAKENHI Grant Numbers JP17K12646, JP20K19739, JP21K17708, and JP21H03397, by JST PRESTO Grant Numbers JPMJPR2122 and JPMJPR20C1, by JST ERATO Grant Number JPMJER2301, by Value Exchange Engineering, a joint research project between Mercari, Inc.\ and the RIISE, by the Singapore Ministry of Education under grant number MOE-T2EP20221-0001, and by an NUS Start-up Grant.
We would like to thank the IJCAI 2023 reviewers and participants, as well as the Games and Economic Behavior reviewers, for their valuable feedback.

\bibliographystyle{plainnat}
\bibliography{main}

\appendix

\section{Quota Constraints}
\label{app:quota}

In some applications of fair division, not all allocations are feasible, that is, there may be constraints on the feasible allocations \citep{Suksompong21}.
For example, as discussed in \Cref{sec:intro}, the balancedness constraint is often desirable when allocating participating players to sports teams, as the rule of the sport may require every team to have a certain number of players.
In this appendix, we consider a broader class of cardinality constraints in which each team~$i$ has an upper quota $z_i$, meaning that the team cannot receive more than $z_i$ participants.
The quota can represent the capacity of the team, for instance, the number of volunteers that the community service club could take or the amount of tasks that the team can possibly perform.
We assume without loss of generality that $z_i$ is an integer in $\{0,1,\dots,m\}$; the unconstrained setting studied in the rest of this paper corresponds to taking $z_i = m$ for all $i\in T$.
To ensure that there is at least one feasible allocation, we also assume that $\sum_{i\in T}z_i \ge m$.
Note that if $\sum_{i\in T}z_i = m$, then each team~$i$ must receive exactly $z_i$ participants.

With quota constraints that may vary from one team to another, envy in the usual sense is sometimes unavoidable---if there are two teams with $z_1 = 2$ and $z_2 = 8$, and $10$ participants of identical positive value are to be allocated, then team~$1$ inevitably has huge envy toward team~$2$.
To handle such heterogeneous constraints, prior work has proposed the notion of \emph{feasible envy} \citep{WuLiGa21,DrorFeSe23}.\footnote{Feasible envy has been applied to more general constraints, including budget constraints \citep{WuLiGa21} and matroid constraints \citep{DrorFeSe23}.}
In our terminology, with nonnegative-value participants, a team~$i$ does not feasible-envy another team~$j$ if for every subset of $j$'s participants that fits within $i$'s quota, $i$'s value for that subset does not exceed her value for her own set of participants.
\citet[Thm.~2]{DrorFeSe23} showed that under one-sided preferences, if all participants yield nonnegative value, the \emph{capped round-robin algorithm} produces a ``feasible EF1'' allocation.
The capped round-robin algorithm is a simple modification of the vanilla round-robin algorithm where a team does not receive any more pick upon reaching its quota.

As it is, the notion of feasible envy does not work well with nonpositive-value participants.
Indeed, if there are two teams with quotas $z_1 = 2$ and $z_2 = 8$, and $10$ participants of identical \emph{negative} value are to be allocated, then team~$2$ inevitably has huge ``feasible envy'' toward team~$1$.
We propose the following definition, which accommodates both nonnegative-value and nonpositive-value participants simultaneously.

\begin{definition}
\label{def:quota-EF1}
In the setting with quotas, an allocation~$A$ is said to satisfy
\begin{itemize}
\item \emph{quota-EF} if for all distinct $i,j\in T$, there exists a subset $B_i\subseteq A_i$ of size $\min\{|A_i|,z_j\}$ such that for every subset $B_j\subseteq A_j$ of size $\min\{|A_j|,z_i\}$, it holds that $v_i(B_i) \ge v_i(B_j)$;
\item \emph{quota-EF1} if for all distinct $i,j\in T$, there exists a subset $B_i\subseteq A_i$ of size $\min\{|A_i|,z_j\}$ such that for every subset $B_j\subseteq A_j$ of size $\min\{|A_j|,z_i\}$, it holds that $v_i(B_i\setminus X)\ge v_i(B_j\setminus Y)$ for some $X\subseteq B_i$ and $Y\subseteq B_j$ with $|X\cup Y|\le 1$;
\item \emph{quota-EF[1,1]} if the same condition holds as for quota-EF1, except that $|X\cup Y|\le 1$ is replaced by $|X|,|Y|\le 1$.
\end{itemize}
\end{definition}

Several remarks on Definition~\ref{def:quota-EF1} are in order.

First, observe that if $z_i\le z_j$, then $|A_i|\le z_i\le z_j$ and thus $\min\{|A_i|,z_j\} = |A_i|$.
Analogously, if $z_i\ge z_j$, then $|A_j|\le z_j\le z_i$ and thus $\min\{|A_j|,z_i\} = |A_j|$.

Second, since the unconstrained setting corresponds to taking $z_i = m$ for all $i\in T$, we have $\min\{|A_i|,z_j\} = |A_i|$ and $\min\{|A_j|,z_i\} = |A_j|$ for all $i,j$, which means that the only choices of $B_i$ and~$B_j$ in Definition~\ref{def:quota-EF1} are $A_i$ and $A_j$, respectively.
Hence, quota-EF1 (resp., quota-EF[1,1]) reduces to EF1 (resp., EF[1,1]) in that setting.
By the remarks after Definition~\ref{def:EF1}, in the unconstrained setting, both quota-EF1 and quota-EF[1,1] coincide with EF1 for nonnegative-value participants as well as for nonpositive-value participants.

Third, suppose that all participants yield nonnegative value to every team.
If $z_i\le z_j$, we require comparing $A_i$ with \emph{every} subset $B_j\subseteq A_j$ of size $\min\{|A_j|,z_i\}$, whereas if $z_i\ge z_j$, we require comparing \emph{some} subset $B_i\subseteq A_i$ of size $\min\{|A_i|,z_j\}$ with $A_j$.
Our definition is thus stronger than the EF1 notion of \citet{WuLiGa21}, which---in the case $z_i\ge z_j$---only requires comparing the set $A_i$ itself with $A_j$.\footnote{\citet{DrorFeSe23} defined their feasible-EF1 notion slightly differently.}

Fourth, one could ask whether the condition ``there exists a subset $B_i\subseteq A_i$'' can be replaced by ``for every subset $B_i\subseteq A_i$''.
As the following example shows, the notion resulting from this modification may not be satisfiable, even with nonnegative-value participants only (or nonpositive-value participants only).

\begin{example}
Consider an instance consisting of $n = 2$ teams with quotas $z_1 = 5$ and $z_2 = 3$, and $m = 8$ participants.
Both teams have identical valuations over the participants---they have value $0$ for four of the participants and $1$ for each of the remaining four participants.

Since all participants yield nonnegative value, quota-EF1 and quota-EF[1,1] coincide.
We show next that if the condition ``there exists a subset $B_i\subseteq A_i$'' in Definition~\ref{def:quota-EF1} is replaced by ``for every subset $B_i\subseteq A_i$'', then quota-EF1 cannot be satisfied in this instance.
Call a participant yielding value~$1$ a ``heavy participant'' and one yielding value $0$ a ``light participant''.
Let us consider the four possibilities.
\begin{itemize}
\item Suppose that team~$2$ receives three light participants, so team~$1$ receives four heavy participants and one light participant. 
If team~$2$ considers the set $A_2$ consisting of three light participants against the set $B_1\subseteq A_1$ consisting of three heavy participants, then team~$2$ envies team~$1$ by more than one participant.
\item Suppose that team~$2$ receives one heavy participant and two light participants, so team~$1$ receives three heavy participants and two light participants. 
If team~$2$ considers the set $A_2$ consisting of one heavy participant and two light participants against the set $B_1\subseteq A_1$ consisting of three heavy participants, then team~$2$ envies team~$1$ by more than one participant.
\item Suppose that team~$2$ receives two heavy participants and one light participant, so team~$1$ receives two heavy participants and three light participants. 
If team~$1$ considers the set $B_1\subseteq A_1$ consisting of three light participants against the set $A_2$ consisting of two heavy participants and one light participant, then team~$1$ envies team~$2$ by more than one participant.
\item Suppose that team~$2$ receives three heavy participants, so team~$1$ receives one heavy participant and four light participants. 
If team~$1$ considers the set $B_1\subseteq A_1$ consisting of three light participants against the set $A_2$ consisting of three heavy participants, then team~$1$ envies team~$2$ by more than one participant.
\end{itemize}
Hence, the stronger version of quota-EF1 cannot be satisfied in this instance.

Observe that if we change the value of the last four participants from $1$ to $-1$ (so the instance consists only of \emph{nonpositive}-value participants), a similar argument still shows that the stronger version of quota-EF1 cannot be satisfied.
\end{example}

Next, we introduce a version of balancedness with respect to quotas.

\begin{definition}
In the setting with quotas, an allocation~$A$ is said to be \emph{quota-balanced} provided that the following holds: for $i,j\in T$, if $|A_i| \le |A_j| - 2$, then $|A_i| = z_i$.
\end{definition}

Quota-balancedness ensures that if a team receives at least two participants fewer than another team, this is because the former team has already reached its quota.
Note that without quotas, quota-balancedness reduces to the balancedness notion studied in the rest of this paper.

We can now state a generalization of \Cref{thm:balanced} in the presence of quota constraints.

\begin{theorem}
\label{thm:balanced-quota}
In the setting with quotas, for any instance, a quota-balanced allocation that satisfies quota-EF[1,1] and swap stability exists and can be computed in polynomial time. 
\end{theorem}

Since quota-EF[1,1] reduces to quota-EF1 for nonnegative-value participants as well as for nonpositive-value participants---for similar reasons as the unconstrained analogs EF[1,1] and EF1---we obtain the following corollary, which generalizes Corollary~\ref{cor:balanced}.

\begin{corollary}
In the setting with quotas, for any nonnegative-value participant instance, a quota-balanced allocation that satisfies quota-EF1 and swap stability exists and can be computed in polynomial time.
The same holds for any nonpositive-value participant instance.
\end{corollary}

To establish \Cref{thm:balanced-quota}, we combine our algorithm from \Cref{thm:balanced} with the capped round-robin algorithm of \citet{DrorFeSe23}.
Define a function $f\colon[m]\rightarrow [n]$ by following the round-robin sequence $1,2,\dots,n,1,2,\dots$, with the exception that if a number~$i$ has already appeared $z_i$ times, then we skip over it in the rest of the sequence.
Since we assume that $\sum_{i\in T}z_i \ge m$, this sequence is well-defined.
Our algorithm is the same as \Cref{alg:swap-stable-rr} except that we apply this new function~$f$; let us refer to it as the ``modified \Cref{alg:swap-stable-rr}''.

Observe that in the output allocation~$A$ of the modified \Cref{alg:swap-stable-rr}, the inequality $|A_i| \le |A_j| - 2$ can occur only if team~$i$ has reached its quota, so $A$ is quota-balanced.
Moreover, swap stability and polynomial running time can be shown as in the unconstrained setting (Lemmas~\ref{lem:swap-stable} and \ref{lem:polynomial-time}).
It therefore remains to prove that $A$ satisfies quota-EF[1,1].

\begin{lemma}
\label{lem:quota-EF1-1}
In the setting with quotas, the output allocation $A$ of the modified Algorithm~\ref{alg:swap-stable-rr} is quota-EF[1,1].
\end{lemma}

\begin{proof}
Fix arbitrary distinct $i,j\in T$.
We consider three cases. 

\underline{Case 1}: $|A_i| = |A_j|$.
Since $|A_i| \le z_i$ and $|A_j| \le z_j$, we have $\min\{|A_i|,z_j\} = |A_i|$ and $\min\{|A_j|,z_i\} = |A_j|$.
This means that the only choices of $B_i$ and $B_j$ in Definition~\ref{def:quota-EF1} are $A_i$ and $A_j$, respectively.
Then, the argument given in Lemma~\ref{lem:EF1-1} for the case $|A_i| = |A_j|$ holds as is.
(If $|A_i| = |A_j| = 0$, EF[1,1] from $i$ to $j$ holds trivially.)

\underline{Case 2}: $|A_i| > |A_j|$.
Since $|A_i|\le z_i$, we have $\min\{|A_j|,z_i\} = |A_j|$.
This means that the only choice of $B_j$ in Definition~\ref{def:quota-EF1} is $A_j$.
We consider two subcases.
\begin{itemize}
\item \underline{Case 2.1}: $|A_i| \ge |A_j| + 2$.
Since $A$ is quota-balanced, it must be that $z_j = |A_j|$, so $\min\{|A_i|,z_j\} = |A_j|$.
Let $B_i$ be the subset of $A_i$ consisting of the $|A_j|$ participants matched to the vertices with the lowest indices in $Q$.\footnote{The set $Q$ is defined as in \Cref{alg:swap-stable-rr}.}
The same argument as in the case $|A_i| = |A_j|$ holds, with $B_i$ taking the role of $A_i$.
\item \underline{Case 2.2}: $|A_i| = |A_j| + 1$.
Since $z_j\ge |A_j|$, we have $\min\{|A_i|,z_j\} = \min\{|A_j|+1,z_j\} \in \{|A_j|, |A_j| + 1\}$.
If $\min\{|A_i|,z_j\} = |A_j|$, the same argument as in Case~2.1 holds.
Else, $\min\{|A_i|,z_j\} = |A_j| + 1$, which means that $z_j \ge |A_j|+1$ and thus $i < j$.
Since $|A_i| = |A_j| + 1$, the only set $B_i\subseteq A_i$ of size $\min\{|A_i|,z_j\} = |A_j|+1$ is $A_i$ itself.
We may now apply the argument given in Lemma~\ref{lem:EF1-1} for the case $|A_i| > |A_j|$.
\end{itemize}

\underline{Case 3}: $|A_i| < |A_j|$.
Since $|A_j|\le z_j$, we have $\min\{|A_i|,z_j\} = |A_i|$.
This means that the only choice of $B_i$ in Definition~\ref{def:quota-EF1} is $A_i$.
We consider two subcases.
\begin{itemize}
\item \underline{Case 3.1}: $|A_i| \le |A_j| - 2$.
Since $A$ is quota-balanced, it must be that $z_i = |A_i|$, so $\min\{|A_j|,z_i\} = |A_i|$.
Let $k = |A_i|$, and consider any subset $B_j\subseteq A_j$ of size~$k$.
Observe that for each $\ell\in\{2,3,\dots,k\}$, the $\ell$th lowest index of a vertex in $Q$ matched to a participant in $B_j$ is higher than the $(\ell-1)$th lowest index of a vertex in $Q$ matched to a participant in $A_i$.
Then, a similar argument as in the case $|A_i| = |A_j|$ holds, with $B_j$ taking the role of $A_j$.

\item \underline{Case 3.2}: $|A_i| = |A_j| - 1$.
Since $z_i\ge |A_i|$, we have $\min\{|A_j|,z_i\} = \min\{|A_i|+1,z_i\} \in \{|A_i|, |A_i| + 1\}$.
If $\min\{|A_j|,z_i\} = |A_i|$, the same argument as in Case~3.1 holds.
Else, $\min\{|A_j|,z_i\} = |A_i| + 1$.
Let $k = |A_i|+1$.
Again, observe that for each $\ell\in\{2,3,\dots,k\}$, the $\ell$th lowest index of a vertex in $Q$ matched to a participant in $A_j$ is higher than the $(\ell-1)$th lowest index of a vertex in $Q$ matched to a participant in $A_i$.
Moreover, since $|A_i| = |A_j|-1$, the only set $B_j\subseteq A_j$ of size $\min\{|A_j|,z_i\} = |A_i|+1$ is $A_j$ itself.
Then, a similar argument as in Lemma~\ref{lem:EF1-1} for the case $|A_i| < |A_j|$ holds.
\end{itemize}

Hence, in all three cases, the allocation~$A$ is quota-EF[1,1].
\end{proof}

\end{document}